\documentclass[12pt]{iopart}

\usepackage[utf8]{inputenc}
\usepackage[T1]{fontenc}
\usepackage{eucal}
\usepackage{microtype}
\usepackage{csquotes}
\expandafter\let\csname equation*\endcsname\relax
\expandafter\let\csname endequation*\endcsname\relax
\usepackage{mathtools}
\usepackage{amssymb,amsfonts,stmaryrd,amsthm,mathtools}
\usepackage{mathrsfs}  
\usepackage{tikz}
\usetikzlibrary{cd}
\usepackage{color,soul}
\usepackage{enumitem}
\usepackage[normalem]{ulem}
\setlist[enumerate,1]{label={\roman*)}}
\usepackage[unicode=true, pdfusetitle, colorlinks=true, allcolors=blue, bookmarks=false]{hyperref}
\usepackage{graphicx,subcaption}

\newcommand{\X}{\mathcal{X}}

\usepackage[left=2.5cm,right=2.5cm,top=3cm,bottom=3.5cm]{geometry}

\theoremstyle{plain}
\newtheorem{theorem}{Theorem}
\newtheorem*{theorem*}{Theorem}

\newtheorem*{lemma*}{Lemma}

\newtheorem*{proposition*}{Proposition}
  
\newtheorem*{corollary*}{Corollary}
\theoremstyle{definition}
\newtheorem{definition}{Definition}
\newtheorem*{definition*}{Definition}
\newtheorem{example}{Example}
\newtheorem*{example*}{Example}
\theoremstyle{remark}
\newtheorem{remark}{Remark}
\newtheorem*{remark*}{Remark}

\newtheorem*{conjecture*}{Conjecture}

\newtheorem*{problem*}{Problem}

\newcommand*{\NN}{\mathbb{N}}
\newcommand*{\FF}{\mathbb{F}}

\newcommand*{\RR}{\mathbb{R}}
\let\R\RR

\newcommand*{\dd}{\mathrm{d}}

\newcommand*{\contr}[1]{\iota_{#1}}
\newcommand*{\liedv}[1]{\mathcal{L}_{#1}}

\DeclareMathOperator{\Image}{Im}
\let\Im\Image

\newcommand{\norm}[1]{\left\lvert\left\lvert #1 \right\rvert\right\rvert}

\newcommand{\hybrid}{\mathscr{H}}
\newcommand{\T}{T}
\newcommand{\cT}{T^{\ast}}
\newcommand{\Cinfty}{C^{\infty}}

\newcommand\restr[2]{{
  \left.\kern-\nulldelimiterspace 
  #1 
  \right|_{#2} 
}}

\let\oldemph\emph
\let\emph\textbf

\usepackage[style=ext-numeric,citestyle=numeric-comp,sorting=nyt,sortcites, backend=biber, giveninits=true, maxnames=99]{biblatex}
\DeclareFieldFormat[article,periodical]{volume}{\mkbibbold{#1}}
\DeclareFieldFormat[article,periodical]{number}{\mkbibparens{#1}}
\renewbibmacro{in:}{%
  \ifentrytype{article}
    {}
    {\printtext{\bibstring{in}\intitlepunct}}}

\DeclareFieldFormat{issuedate}{#1}
\renewcommand{\jourvoldelim}{\addcomma\space}

\renewbibmacro*{journal+issuetitle}{%
  \usebibmacro{journal}%
  \setunit*{\jourvoldelim}%
  \iffieldundef{series}
    {}
    {\setunit*{\jourserdelim}%
     \printfield{series}%
     \setunit{\servoldelim}}%
  \usebibmacro{volume+number+eid}%
  \setunit{\addcolon\space}%
  \usebibmacro{issue}%
  \setunit{\bibpagespunct}%
  \printfield{pages}%
  \setunit{\space}%
  \usebibmacro{issue+date}%
  \newunit}

\renewbibmacro*{note+pages}{%
  \printfield{note}%
  \newunit}

\DeclareFieldFormat[article,periodical]{date}{\mkbibparens{#1}}

\renewbibmacro*{publisher+location+date}{%
  \printlist{publisher}%
  \iflistundef{location}
    {\setunit*{\addcomma\space}}
    {\setunit*{\addcolon\space}}%
  \printlist{location}%
  \setunit*{\addcomma\space}%
  \usebibmacro{date}%
  \newunit}

\AtEveryBibitem{\clearfield{month}}
\AtEveryBibitem{\clearfield{day}}

\AtEveryBibitem{\clearfield{urldate}}

\AtEveryBibitem{
  \ifentrytype{online}{}{
    \ifentrytype{thesis}{}{
      \ifentrytype{unpublished}{}{ 
        \clearfield{url}
        \clearfield{urldate}
      }
    }
  }
  \ifentrytype{unpublished}{}{ 
  \clearfield{note}
  }
  \clearfield{language}
}

\DeclareSourcemap{
  \maps[datatype=bibtex]{
    \map[overwrite=true]{
      \step[fieldset=urldate, null]
      \step[fieldset=language, null]
      \step[fieldset=address, null]
     \step[fieldset=pagetotal, null]
    }
  }
}

\DeclareSourcemap{
    \maps[datatype=bibtex]{
      \map{
        \step[fieldsource=eprint,final]
        \step[fieldset=url,null]
        \step[fieldset=doi,null]
      }  
    }
  }

\DeclareSourcemap{
    \maps[datatype=bibtex]{
      \map{
        \step[fieldsource=doi,final]
        \step[fieldset=url,null]
      }  
    }
  }

\AtEveryBibitem{\clearfield{pubstate}} 
\AtEveryBibitem{\clearfield{version}} 

\mathtoolsset{showonlyrefs}

\allowdisplaybreaks 
\begin{document}


\title[Hamilton--Jacobi theory for nonholonomic and forced hybrid systems]{Hamilton--Jacobi theory for nonholonomic and forced hybrid mechanical systems}

\author{Leonardo Colombo$^{1}$, Manuel de Le\'on$^{2,3}$, Mar\'ia Emma Eyrea Iraz\'u$^4$ and Asier L\'opez-Gord\'on$^{2}$\footnote{Author to whom correspondence should be addressed.}}

\date{\today}

\address{$^{1}$Centro de Automática y Robótica (CSIC-UPM), 
Carretera de Campo Real, km 0, 200, 28500 Arganda del Rey, Spain.\\$^{2}$Instituto de Ciencias Matemáticas (CSIC-UAM-UC3M-UCM) 
Calle Nicolás Cabrera, 13-15, Campus Cantoblanco, UAM, 28049 Madrid, Spain.\\ $^{3}$Real Academia de Ciencias Exactas, Físicas y Naturales Calle Valverde, 22, 28004, Madrid, Spain.\\
$^{4}$CONICET-CMaLP-Department of Mathematics, Universidad Nacional de La Plata. 
Calle 1 y 115, La Plata 1900, Buenos Aires, Argentina
}
\ead{leonardo.colombo@car.upm-csic.es, mdeleon@icmat.es, maeemma@mate.unlp.edu.ar, asier.lopez@icmat.es}






\vspace{10pt}

\begin{abstract}
A hybrid system is a system whose dynamics is given by a mixture of both continuous and discrete transitions. {In particular, these systems can be utilised to describe the dynamics of a mechanical system with impacts.}
Based on the approach by Clark \cite{Clark2020}, we develop a geometric Hamilton--Jacobi theory for forced and nonholonomic hybrid dynamical systems.
We state the corresponding Hamilton--Jacobi equations for these classes of systems and apply our results to analyze some examples. 

\medskip


\medskip

\noindent{\bf MSC\,2020 codes:} 70F35, 93C30 (primary); 53Z05, 70F40, 70H20 (secondary)

\end{abstract}


%
%
%
%
%


\section{Introduction}

As it is well-known, in classical mechanics, Hamilton--Jacobi theory is a way to integrate a
system of ordinary differential equations (Hamilton equations) that, through an appropriate
canonical transformation, is led to equilibrium. The equation to be satisfied by the
generating function of this transformation is a partial differential equation whose solution allows
us to integrate the original system. In this respect, Hamilton--Jacobi theory provides important physical examples of the deep connection between first-order partial differential equations and systems of first-order ordinary differential equations. 
Moreover, Hamilton--Jacobi theory provides a remarkably powerful method to integrate the dynamics of many Hamiltonian systems. In particular, for a completely integrable system, knowing a complete solution of the Hamilton--Jacobi problem, the dynamics of the system can be easily reduced to quadratures (see \cite{Grillo2021, Grillo2021a, Grillo2016}).
For these reasons, Hamilton--Jacobi theory has been a matter of
continuous interest and has been studied classically as well as in other contexts \cite{Dominici1984}.

From the viewpoint of geometric mechanics, the intrinsic formulation of Hamilton--Jacobi equation is also clear. Nevertheless, in the papers \cite{Carinena2006, Iglesias-Ponte2008} a new geometric framework for
the Hamilton--Jacobi theory was presented and the Hamilton--Jacobi equation was formulated
both in the Lagrangian and in the Hamiltonian formalisms of autonomous and non-autonomous
mechanics. A similar generalization of the Hamilton--Jacobi formalism was outlined in \cite{Krupkova}. {This geometric framework was used to extend the Hamilton--Jacobi theory for several contexts (see \cite{Esen2022} and references therein).}

The aim of this paper is to go ahead with this program and develop the geometric description of the Hamilton--Jacobi theory for nonholonomic and forced {mechanical systems with impacts, modeled as} hybrid dynamical systems.
{A Hamilton--Jacobi for hybrid Hamiltonian systems has been studied by Clark \cite{Clark2020}, whose approach we have extended to forced and nonholonomic hybrid systems.}
Hybrid systems are dynamical systems with continuous-time and discrete-time components in their dynamics. This class of dynamical systems is capable of modeling several physical systems, such as multiple UAV (unmanned aerial
vehicles) systems \cite{Lee2013} and legged robots  \cite{Westervelt2018}, among many others \cite{Goebel2012, goodman2019existence, vanderSchaft2000}. Simple hybrid systems are a class of hybrid systems introduced in \cite{Johnson1994}, denoted as such because of their simple structure. A simple hybrid system is characterized by a tuple $\hybrid=(\mathcal{X}, S, z, \Delta)$, where $\mathcal{X}$ is a smooth manifold, $z$ is a smooth vector field on $\mathcal{X}$, $S$ is an embedded submanifold of $\mathcal{X}$ with co-dimension $1$, and $\Delta:S\to \mathcal{X}$ is a smooth embedding. This type of hybrid system has been mainly employed for the understanding of locomotion gaits in bipeds and insects \cite{colombo2020symmetries, Holmes2006, Westervelt2018}, and are of great interest within the optimal control community \cite{Clark2019,Clark2021,Clark2021a,Clark2023,Goebel2012,vanderSchaft2000}.
{Moreover, if $\mathcal{X}$ is the tangent (resp.~cotangent) bundle of a manifold and $\Delta$ does not change the base point then the hybrid system can be used to model a mechanical system with impacts. In that case, $\Delta$ represents the instantaneous changes of velocity (resp.~momentum) that occur at the impacts. It is worth mentioning that one could consider more general hybrid systems, for instance, systems for which the Hamiltonian function or the symplectic structure are modified at certain instants. Nevertheless, such generalizations will not be necessary in order to model mechanical systems with impacts.}

The remainder of the paper is structured as follows. After establishing some basic concepts and notation about forced Hamiltonian systems, hybrid dynamical systems and the geometric Hamilton--Jacobi theory in Section \ref{sec2}, we introduce the geometric Hamilton--Jacobi theory for hybrid dynamical systems in Section \ref{sec3}.
{We show a Hamilton--Jacobi theorem for hybrid systems (Theorem \ref{HJ_theorem_hybrid}), based on Theorem VIII.3 from \cite{Clark2020}. This result is subsequently applied}
in Section \ref{sec4} to forced hybrid mechanical systems (Theorem \ref{HJ_theorem_hybrid_forced}) and to nonholonomic hybrid mechanical systems (Theorem \ref{HJ_theorem_hybrid_nh}) in Section \ref{sec:nonholonomic}. Several examples are shown in the paper. {Finally, in Section~\ref{sec:conclusions} we mention some possible future lines of research opened up by the results of this paper.}

\section{Preliminaries on hybrid dynamical systems and Hamilton--Jacobi theory}\label{sec2}
We begin by describing the class of systems we will consider in this work as well as reviewing the Hamilton--Jacobi theory for Hamiltonian mechanical systems.

\subsection{Forced Hamiltonian systems}
Let $Q$ be an $n$-dimensional manifold, representing the space of positions of a mechanical system. Let $T^\ast Q$ denote its cotangent bundle, with canonical projection $\pi_Q: T^\ast Q\to Q$. If $Q$ has local coordinates $(q^i)$, then $T^\ast Q$ has induced bundle coordinates $(q^i, p_i)$. As it is well-known, the cotangent bundle is endowed with a canonical one-form $\theta_Q= p_i \dd q^i$, which defines a canonical symplectic form $\omega_Q=-\dd \theta_Q=\dd q^i \wedge \dd p_i$. 

For each function $H$ on $T^\ast Q$, the symplectic form $\omega_Q$ defines a unique vector field $X_H$ on $T^\ast Q$ given by    $\omega_Q\left(X_H, \cdot \right) = \dd H$, which is called the \emph{Hamiltonian vector field} of $H$. Locally,
\begin{equation}
    X_H = \frac{\partial H}{\partial p_i} \frac{\partial}{\partial q^i} - \frac{\partial H}{\partial q^i} \frac{\partial}{\partial p_i}, 
\end{equation}
so its integral curves are given by \textit{Hamilton's equations}
\begin{equation}
\frac{\dd q^i}{\dd t} = \frac{\partial H}{\partial p_i},\quad \frac{\dd p_i}{\dd t} = - \frac{\partial H}{\partial q^i}.
\end{equation}
Therefore, the dynamics of a mechanical system on $Q$ with the Hamiltonian function $H\colon T^\ast Q \to \R$ can be characterized by the Hamiltonian vector field $X_H$.

Consider now a forced Hamiltonian system $(H, F)$, with $H$ a Hamiltonian function on $T^\ast Q$ and $F=F_i(q, p) \dd q^i$ a semibasic one-form on $T^\ast Q$ representing an external (non-conservative) force. Then, we can define the vector field $Z_F$ by $\omega_Q\left(Z_F, \cdot\right) = F$, or, locally, $Z_F = -F_i \frac{\partial}{\partial p_i}$.

We also define the \emph{forced Hamiltonian vector field} $X_{H,F} = X_H + Z_F$, locally given by
\begin{equation}
    X_{H,F}= \frac{\partial H}{\partial p_i} \frac{\partial}{\partial q^i} - \left(\frac{\partial H}{\partial q^i} + F_i\right) \frac{\partial}{\partial p_i}, 
\end{equation}
so that its integral curves are given by the \textit{forced Hamilton's equations}
\begin{equation}
 \frac{\dd q^i}{\dd t} = \frac{\partial H}{\partial p_i},\quad \frac{\dd p_i}{\dd t} = - \frac{\partial H}{\partial q^i} - F_i.
\end{equation}

\subsection{Hybrid dynamical systems}
Hybrid dynamical systems are dynamical systems characterized by their mixed behavior of continuous and discrete dynamics where the transition is determined by the time when the continuous flow switches from the ambient space to a co-dimensional one submanifold. This class of dynamical systems is given by a $4$-tuple $\hybrid=(\X,S,z,\Delta)$. The pair ($\X$,$z$) describes the continuous dynamics as \begin{equation*}
\dot{x}(t) = z(x(t))
\end{equation*} 
where $\X$ is a smooth manifold and $z$ a $C^1$ vector field on $\X$ with flow $\varphi_t:\X\rightarrow\X$. Additionally, ($S$,$\Delta$) describes the discrete dynamics as 
$x^{+}=\Delta(x^{-})$ where $S\subset\X$ is a smooth submanifold of co-dimension one called the \emph{impact surface}.

The hybrid dynamical system describing the combination of both dynamics is given by
\begin{equation}\label{sigma}\Sigma_{\hybrid}: \begin{cases}
\dot{x} = z(x),& x\not\in S\\
x^+ = \Delta(x^-),& x^-\in S.
\end{cases}\end{equation}
The flow of the hybrid dynamical system \eqref{sigma} is denoted by $\varphi_t^{\hybrid}$. This may cause a little confusion around the break points, that is, where $\varphi_{t_0}(x)\in S$. It is not clear whether $\varphi_{t_0}^{\hybrid}(x) = \varphi_{t_0}(x)$ or $\varphi_{t_0}^{\hybrid}(x) = \Delta(\varphi_{t_0}(x))$. In other words, if the state at the time of impact with the submanifold $S$ is $x^-$ or $x^+$. We will take the second value, i.e. $\varphi_{t_0}^{\hybrid}(x)=x^+$.

A solution of a hybrid dynamical system may experience a Zeno state if infinitely many impacts occur in a finite amount of time. To exclude these types of situations, we require the set of impact times to be closed and discrete. {As in \cite{Westervelt2018}, we will assume implicitly throughout the remainder of the paper that $\overline{\Delta}({S})\cap{S}=\emptyset$, where $\overline{\Delta}({S})$ denotes the closure of $\Delta({S})$.} 
\begin{definition}
    Let $\hybrid=(\X,S,z,\Delta)$ be a hybrid dynamical system.
    A function $f$ on $\mathcal X$ is called a \emph{hybrid constant of the motion} if it is preserved by the hybrid flow, namely, {$f \circ \varphi_t^\hybrid = f$.} In other words,
    $z(f)=0$ and $ f \circ \Delta = f \circ i$, where $i\colon S \hookrightarrow \X$ is the canonical inclusion.
\end{definition}

	A hybrid dynamical system $(\X,S,z,\Delta)$ is said to be a \emph{hybrid Hamiltonian system} if it is determined by {$\hybrid_H\coloneqq (T^{*}Q, {S_H}, X_{H}, \Delta_H)$}, where $X_{H}:T^{*}Q\to T(T^{*}Q)$ is the Hamiltonian vector field associated with the Hamiltonian function $H$, 
{${S_H}$ is the switching surface, a submanifold of $T^{*}Q$ such that $\pi_Q(S_H)$ is a codimension-one submanifold of $Q$, and $\Delta_H:{S}_H\to T^{*}Q$ is the impact map, a smooth embedding that preserves the base point, namely, $\pi_Q\circ \Delta_H = \pi_Q$, where $\pi_Q\colon T^{*}Q \to Q$ denotes the canonical projection of the cotangent bundle. The dynamics may be restricted to a \emph{domain} $\X$, that is, a submanifold $\X = \cT U\subseteq \cT Q$, where $U\subseteq Q$ is an open subset.}

  The dynamical system generated by $\hybrid_H$ is given by
  \begin{equation}
    \label{RHDS}\Sigma_{\hybrid_H}:
    \left\{\begin{array}{ll}\dot{x}(t)=X_{H}(x(t)), & \hbox{ if } x(t)\notin{S_H},\\ x^{+}(t)=\Delta_H(x^{-}(t)),&\hbox{ if } x^-(t)\in{S_H}, 
    \end{array}\right.
  \end{equation}
  where $x(t)=(q(t),{p}(t))\in T^{*}Q$, 
  and $x^{-}$, $x^{+}$ denote the states immediately before and after the times when $x(t)$ intersect with $S_H$, namely 
  $$x^{-}(t)\coloneqq \displaystyle{\lim_{\tau\to t^{-}}}x(\tau), \qquad 
  x^{+}(t)\coloneqq \displaystyle{\lim_{\tau\to t^{+}}}x(\tau)\, .$$


Similarly, a hybrid system is called a \emph{forced hybrid Hamiltonian system} if it is determined by $\hybrid_{H, F}\coloneqq (T^{*}Q, {S_H}, X_{H,F}, \Delta_H)$, where $X_{H,F}:T^{*}Q\to T(T^{*}Q)$ is the forced Hamiltonian vector field associated with the forced Hamiltonian system $(H,F)$. {As in the unforced case, it is assumed that $\pi_Q(S_H)$ is a codimension-$1$ submanifold of $Q$ and $\pi_Q \circ \Delta_H = \pi_Q$.}

Physically, $\pi_Q(S_H)$ can be usually regarded as the wall where the impact occurs, while $S_H$ is the subspace of the phase space in which positions are in the wall and momenta are pointing to the wall. The impact map represents the change of momenta produced in the instant of the impact, e.g., an elastic collision with a wall. In many examples, the switching surface and the impact map are determined as follows.

\begin{remark}[Newtonian impact law]\label{remark:Newtonian_Hamiltonian_impact_law}
    Consider a mechanical Hamiltonian function $H\in \Cinfty(\cT Q)$, namely,
     \begin{equation}
         H (q, p) = g^{-1}_q(p,p) - V(q)\, ,
     \end{equation}
     where $g$ is a Riemannian metric, $g_q^{-1}$ denotes the inverse of $g$ at $q$ and $V$ is a function on $Q$. Let $h\in \Cinfty(Q)$ be a function such that $0$ is a regular value, that is, $h^{-1}(0)$ is a submanifold of $Q$, called the \emph{constraint function}. These functions define a hybrid Hamiltonian system where
     \begin{enumerate}
         \item the domain is $\X = \left\{(q, p)\in \cT Q  \mid h (q) \geq 0 \right\}$,
         \item $X_H$ is the Hamiltonian vector field of $H$,
         \item the switching surface is
         \begin{equation}
             S = \left\{(q, p)\in \cT Q \mid h(q) = 0\, , \ \langle\langle p, \dd h_q\rangle\rangle_q < 0 \right\}\, ,
         \end{equation}
         \item the impact map is $\Delta (q, p) = \big(q, P_q(p)\big)$, with $P_q \colon \cT_q Q \to \cT_q Q$ being the map given by
         \begin{equation}
             P_q (p) = p - (1+e) \frac{\langle\langle p, \dd h_q\rangle\rangle_q}{\norm{\dd h_q}_q^2} \dd h_q\, ,
         \end{equation}
     \end{enumerate}
     where $\norm{\cdot}_q$ and $\langle\langle \cdot, \cdot\rangle\rangle_q$ denote the norm and the inner product defined by $g$ on $\cT Q$, and $e\in [0,1]$ is a constant called the \emph{coefficient of restitution}. In particular, $e=1$ and $e=0$ correspond to purely elastic and purely inelastic impacts, respectively. Of course, by replacing $X_H$ with $X_{H,F}$, a hybrid forced Hamiltonian system may also be described by this law.
 \end{remark}

{
\begin{remark}
    Although this is not required by any of the results in the paper, it is possible to assume that the impact map ``preserves the symplectic structure'' in the sense that $\Delta_H^\ast \omega_Q = \restr{\omega_Q}{S_H}$, then the impact map must take the form
    \begin{equation}
        \Delta_H(q, p) = \big(q, p + f \dd h_q\big)\, ,
    \end{equation}
    for a function $f$ on $\cT Q$ and a function $h$ on $Q$. In particular, this condition is satisfied by the examples on the paper.
\end{remark}
}

\subsection{Geometric Hamilton--Jacobi theory}

Consider a Hamiltonian function $H\colon T^\ast Q \to \R$.
In the standard formulation, the Hamilton--Jacobi problem consists on finding a function ${\mathcal{S}}\colon \R\times Q \to \R$, known as the principal function, such that the partial differential equation (PDE) 
\begin{equation} \label{1}
    \frac{\partial {\mathcal{S}}}{\partial t}+H\left(q,\frac{\partial {\mathcal{S}}}{\partial q}\right)=0
\end{equation}
is satisfied. If we assume ${\mathcal{S}}(t,q)=W(q)-tE,$ where $E$ is a constant, then the function $W$, known as the characteristic function, has to satisfy
\begin{equation} 
    H\left(q,\frac{\partial W}{\partial q}\right)=E.
    \label{HJ_classical}
\end{equation}
Both of the above PDEs are known as the Hamilton--Jacobi equation.

Hamilton--Jacobi theory can be given the following geometrical interpretation \cite{Abraham2008,Carinena2006}. Given the Hamiltonian vector field $X_H$ of $H$, we want to find a section $\gamma$ of $\pi_Q: T^\ast Q \to Q$ which maps integral curves of the projected vector field $X_H^\gamma= T \pi_Q \circ X_H \circ \gamma$ {on $Q$} into integral curves of $X_H$. In other words, we look for a one-form $\gamma$ on $Q$ such that
\begin{equation}
\label{gamma_related_int_curves}
X_H^\gamma\circ \sigma(t)=\frac{\mathrm{d} } {\mathrm{d}t}\sigma(t) \Longrightarrow  X_H\circ( \gamma\circ \sigma(t))=\frac{\mathrm{d} } {\mathrm{d}t}\left(\gamma\circ\sigma (t)\right)  ,  
\end{equation}
{for every integral curve $\sigma \colon I \subseteq \R \to Q$ of $X_H^\gamma$.}
The one-form $\gamma$ satisfies the condition \eqref{gamma_related_int_curves} if, and only if, the vector fields $X_H^\gamma$ and $X_H$ are $\gamma$-related, that is,
\begin{equation}
    X_H\circ \gamma=T\gamma\circ X_H^\gamma.
    \label{X_gamma_related}
\end{equation}
In other words, the following diagram commutes: 
\begin{center}
\begin{tikzcd}
T^\ast Q \arrow[rr, "X_{H}"] \arrow[d, "\pi_{Q}"] &  & T T^\ast Q \arrow[d, "T\pi_{Q}"']    \\
Q \arrow[rr, "X_H^{\gamma}"'] \arrow[u, "\gamma", bend left, shift left]    &  & TQ \arrow[u, "T\gamma"', bend right, shift right]
\end{tikzcd}
\end{center}
Locally,
\begin{equation}
     X_H\circ \gamma = \frac{\partial H} {\partial p_i  } \frac{\partial  } {\partial q^i} - \frac{\partial H} {\partial q^i} \frac{\partial  } {\partial p_i},
\end{equation}
and
\begin{equation}
    T\gamma\circ X_H^\gamma = \frac{\partial H} {\partial p_i  } \frac{\partial  } {\partial q^i} + \frac{\partial H} {\partial p_j} \frac{\partial \gamma_i } {\partial q^j} \frac{\partial}{\partial p_i},
\end{equation}
{where the right-hand sides of these equations are evaluated along $\gamma$}, so equation~\eqref{X_gamma_related} holds if and only if
\begin{equation}
    - \frac{\partial H} {\partial q^i} = \frac{\partial H} {\partial p_j} \frac{\partial \gamma_i} {\partial q^j}.
\end{equation}

In addition, if $\gamma$ is closed, we have that
\begin{equation}
    \frac{\partial \gamma_i} {\partial q^j} = \frac{\partial \gamma_j} {\partial q^i}, 
\end{equation}
so
\begin{equation}
    - \frac{\partial H} {\partial q^i} = \frac{\partial H} {\partial p_j} \frac{\partial \gamma_j} {\partial q^i},
\end{equation}
that is,
\begin{equation}
    0 = \left(  \frac{\partial H} {\partial q^i}  + \frac{\partial H} {\partial p_j} \frac{\partial \gamma_j} {\partial q^i} \right) \dd q^i 
    = \dd \left(H \circ \gamma  \right).
    \label{HJ_differential}
\end{equation}
{Since $\gamma$ is closed, by the Poincaré Lemma it is locally exact. More specifically, around each point $q\in Q$ there exists an open subset $U\subseteq Q$ such that $\restr{\gamma}{U} = \dd W$ for a function $W\in \Cinfty(U)$. Hence, equation~\eqref{HJ_differential} can be locally written as equation~\eqref{HJ_classical}.}


Additionally, if a one-form $\gamma$ on $Q$ satisfies equation~\eqref{X_gamma_related}, then $X_H\circ \gamma$ is contained in the tangent space of $\Im \gamma$.

Combining all the above, we have proven the following:

\begin{theorem}[Hamilton--Jacobi theorem] \label{HJ_theorem}
Consider a Hamiltonian function $H$ on $T^\ast Q$.
Let $\gamma$ be a closed one-form and let $X_H^\gamma = T \pi_Q \circ X_H \circ \gamma$. Then, the following assertions are equivalent:
\begin{enumerate}
    \item If $\sigma:\mathbb{R}\rightarrow Q$ is an integral curve of $X_{H}^{\gamma}$ then $\gamma\circ\sigma$ is an integral curve of $X_{H}.$
     \item $\dd(H\circ \gamma)=0$, 
    \item {$X_H$ is tangent to $\Im \gamma$.}
\end{enumerate}
A one-form $\gamma$ satisfying the statements above is called a \emph{solution of the Hamilton--Jacobi problem} for $H$.
\end{theorem} 

\begin{definition}\label{def_complete_sols}
Consider a solution $\gamma_{\lambda}$ of the Hamilton--Jacobi problem for $H$ depending on $n$ additional parameters $\lambda\subset \mathbb{R}^n $. If the map $\Phi:Q\times \mathbb{R}^n\rightarrow T^*Q$ given by $\Phi_\lambda(q)=\Phi(q,\lambda_1,...,\lambda_n)=\gamma_{\lambda}(q)$ is a local diffeomorphism, we say that $\gamma_{\lambda}(q)$ is a \emph{complete solution of the Hamilton--Jacobi problem} for $H$.
\end{definition}

\begin{remark}\label{remark:Liouville_integrability}
    A complete solution of the Hamilton--Jacobi problem for $H$ produces a foliation on Lagrangian submanifolds invariant under the Hamiltonian flow. As a matter of fact, if $\Phi\colon Q \times \R^n \to T^\ast Q$ is a complete solution, then the $n$ functions $f_a= \pi_a \circ \Phi^{-1}\colon T^\ast Q\to \R^n$ (where $\pi_a\colon Q\times \R^n \ni (q, \lambda_1, \ldots, \lambda_a\ldots, \lambda_n)\mapsto \lambda_a\in \R$ denotes the natural projection) are independent constants of the motion in involution. 
    {If a Hamiltonian system has a set of $n$ constants of the motion which are in involution and functionally independent almost everywhere, then it is called}
    completely integrable (in the sense of Liouville).
    {For more details on the notion of Liouville integrability, the reader may consult Section 1.4 in \cite{Bolsinov2004}.}
    Furthermore, complete solutions may be employed to compute action-angle coordinates \cite{Arnold1978a}.
\end{remark}

Theorem \ref{HJ_theorem} can be generalized to forced Hamiltonian systems as follows (see \cite{deLeon2022c}).

\begin{theorem}[Forced Hamilton--Jacobi theorem] \label{teoFor}
Let $(H,F)$ be a forced Hamiltonian system on $T^*Q$ and let $\gamma$ be a closed one-form on $Q$. Then the following conditions are equivalent:
\begin{enumerate}
    \item $\dd(H\circ\gamma)=-\gamma^*F$
    \item If $\sigma:\mathbb{R}\rightarrow Q$ is an integral curve of $X_{H,F}^{\gamma}$ then $\gamma\circ\sigma$ is an integral curve of $X_{H,F}.$
    \item $X_{H,F}$ is tangent to $\Im \gamma$.
\end{enumerate}
If $\gamma$ satisfies these conditions, it is called a \emph{solution of the Hamilton--Jacobi problem} for $(H,F).$
\end{theorem}

{In order to show this theorem, observe that $X_{H, F}^\gamma = X_H^\gamma$ and 
\begin{equation}
    X_{H, F} \circ \gamma = X_H \circ \gamma -  F_i \frac{\partial}{\partial p_i}\, .
\end{equation}
along $\gamma$. The rest of the proof is analogous to the one of Theorem \ref{HJ_theorem}.}

\begin{remark}\label{remark:Liouville_integrability_forced}
    Complete solutions $\Phi\colon Q\times \RR^n \to T^\ast Q$ of the Hamilton--Jacobi problem for $(H, F)$ are defined analogously to the unforced case (see Definition~\ref{def_complete_sols}). In \cite{deLeon2022c} we showed that, for each $\lambda\in \R^n$, $\Im \Phi(\cdot, \lambda)$ is a Lagrangian submanifold invariant under the flow of $X_{H, F}$ (see Theorem 4 in \cite{deLeon2022c}). Moreover, we showed that from a complete solution one can construct $n$ constants of the motion in involution, 
    {${f_i =\pi_i \circ \Phi^{-1}} \colon T^\ast Q \to \RR,\, i=1, \ldots, n$} {(see Proposition 8 in \cite{deLeon2022c})},
    {namely, $X_{H,F}(f_i)=0$.}
    {Thus, obtaining a complete solution of the Hamilton--Jacobi problem simplifies solving the equations of motion for $(H, F)$. For instance, one can change the canonical coordinates $(q^i, p_i)$ to $(x_i, y_i)$, with $x_i=f_i,\, i=1, \ldots, n$. However, since $X_{H,F}$ is not, in general, a Hamiltonian vector field, the Liouville--Arnol'd theorem \cite{Arnold1978a, Audin2004, Bolsinov2004, Wiggins2003} cannot be applied to construct action-angle coordinates. In future works, we plan to study the existence of action-angle-like coordinates for forced Hamiltonian systems.}
\end{remark}

\section{Geometric Hamilton--Jacobi theory for hybrid dynamical systems}\label{sec3}

\begin{definition}\label{def:solution_hybrid_HJ}
Consider a hybrid Hamiltonian system $\hybrid_{H}= (T^{*}Q, {S_H}, X_{H}, \Delta_H)$. Let $U_k\subset Q\setminus \pi_Q(S_H)$ be open subsets of $Q$. Let $\mathfrak{U} = \left\{(U_k,\gamma_k)  \right\}$ be a family of closed one-forms $\gamma_k$ on $U_k${, where $k\in \NN$ is a discrete index}. We will say that $\mathfrak{U}$ is a \emph{solution of the hybrid Hamilton--Jacobi problem} for $\hybrid_{H}$ if:
\begin{enumerate}
\item For each $k$, $\gamma_k$ is a solution of the Hamilton--Jacobi equation on $U_k$, namely,
\begin{equation}
    \dd (H \circ \gamma_k) = 0. 
    \label{HJ_eq_hybrid}
\end{equation}
\item If $x\in \partial U_k \cap \pi_Q(S_H)$, then there exists an {index} $l$ such that $x\in \partial U_l \cap \pi_Q(S_H)$ and $\gamma_l$ and $\gamma_k$ are $\Delta_H$-related, i.e.,
\begin{equation}
    \lim_{y\to x} \gamma_l(y) = \Delta_H \left( \lim_{y\to x} \gamma_k(y) \right).
    \label{Delta_related_condition}
\end{equation}
The indices $k$ and $l$ may be different or not. This will depend on the impact map of the problem.
\end{enumerate}
{ 
Unlike continuous systems, where a solution of the Hamilton--Jacobi equation is a single one-form $\gamma$, in the hybrid Hamilton--Jacobi equation a family of one-forms $\gamma_k$ is required. Each one-form $\gamma_k$ will determine the dynamics between the $k$-th and $(k+1)$-th impacts (regarding the initial conditions as the ``$0$-th impact''). For simplicity's sake, along the paper we assume that Zeno effect does not occur so that the impacts are discrete and $k\in \NN$. Nevertheless, one could consider systems that experience Zeno (e.g., the bouncing ball) by replacing the integer index $k$ with a real parameter. It is also worth remarking that the open subsets $U_k$ could have non-empty intersection. As a matter of fact, it is possible to have $U_k = Q \setminus \pi_Q(S_H)$ for all $k$.
}

\end{definition}



\begin{remark}\label{remark:def_Clark}
Clark's definition \cite[Definition VIII.1]{Clark2020} additionally requires the condition 
\begin{equation}
    \lim_{y\to x} \mathcal{A}_l (y)  = \lim_{y\to x} \mathcal{A}_k (y),  
\end{equation}
{where $\mathcal{A}_k\in \Cinfty(U_k)$ and $\mathcal{A}_l\in \Cinfty(U_l)$ are functions such that $\gamma_k = \dd \mathcal A_k$ and $\gamma_l = \dd \mathcal A_l$}. This condition is important for optimal control purposes but irrelevant for ours, so we will not require it.
\end{remark}

The geometric interpretation of the hybrid Hamilton--Jacobi problem is as follows. Given a collection of open subsets $U_k\subset Q \setminus \pi_Q(S_H)$ and closed one-forms $\gamma_k$ on $U_k$, one can project the vector fields $\restr{X_H}{T^\ast U_k}\in \mathfrak{X}(T^\ast U_k)$ along $\gamma_k(U_k)$, obtaining the vector fields $X_H^{\gamma_k}\coloneqq T \pi_Q \circ X_H \circ \gamma_k \in \mathfrak{X}(U_k)$. Additionally, we assume that $\restr{X_H}{T^\ast U_k}$ and $X_H^{\gamma_k}$ are $\gamma_k$-related, i.e., $T\gamma_k \circ X_H^{\gamma_k} =  \restr{X_H}{T^\ast U_k}\circ \gamma$. In other words, the following diagram commutes:
\begin{center}
\begin{tikzcd}
T^\ast U_k \arrow[rr, "\restr{X_H}{T^\ast U_k}"] \arrow[d, "\restr{\pi_Q}{ T^\ast U_k}"] &  & T T^\ast U_k \arrow[d, "\restr{T\pi_Q}{T^\ast U_k}"']    \\
U_k \arrow[rr, "X_H^{\gamma_k}"'] \arrow[u, "\gamma_k", bend left, shift left]    &  & TU_k \arrow[u, "T\gamma_k"', bend right, shift right]
\end{tikzcd}
\end{center}
Therefore, $\gamma_k$ is a solution of the Hamilton--Jacobi problem for $\restr{H}{\cT U_k}$. Moreover, the one-forms $\gamma_l$ and $\gamma_k$ have to be $\Delta_H$-related on $\partial U_k \cap \partial U_l \cap \pi_Q(S_H)$, i.e., they have to satisfy equation~\eqref{Delta_related_condition}.

\begin{theorem}[Hybrid Hamilton--Jacobi theorem]\label{HJ_theorem_hybrid}
Let $\hybrid_{H}= (T^{*}Q, {S_H}, X_{H}, \Delta_H)$ be a hybrid Hamiltonian system, $U_k\subset Q\setminus \pi_Q(S_H)$ be open subsets of $Q$ and $\mathfrak{U} = \left\{(U_k,\gamma_k)  \right\}$ be a family of closed one-forms $\gamma_k\in \Omega^1(U_k)$.  Then, the following statements are equivalent:
\begin{enumerate}
    \item The family $\mathfrak{U}$ {is a solution of} the hybrid Hamilton--Jacobi equation for $\hybrid_{H}$. 
    {
    \item For every curve $c\colon \R\to Q$ such that 
    \begin{enumerate}
        \item $c\big((t_k, t_{k+1})\big)\subset U_k$, 
        \item $c$ intersects $\pi_Q(S_H)$ at $\{t_k\}_k$,
        \item $c$ satisfies the equations
    \begin{align}
    &\dot c(t) = T \pi_Q \circ X_H \circ \gamma_k \circ c(t), \qquad t_k<t<t_{k+1},\\
    & \gamma_{k+1} \circ c(t_{k+1}) = \Delta_H \circ \gamma_k \circ c(t_{k+1}),
    \end{align}
    \end{enumerate}
    the curve $\tilde{c}\colon \RR \to T^\ast Q$ given by $\tilde{c}(t) = \gamma_k \circ c(t)$ for $t\in[t_k, t_{k+1})$ is an integral curve of the hybrid dynamics.}
\end{enumerate}
\end{theorem}

{
This theorem, with the additional assumption from Remark~\ref{remark:def_Clark}, was proven in \cite{Clark2020} (Theorem VIII.3).
}

\begin{proof}
    By Theorem \ref{HJ_theorem}, equation~\eqref{HJ_eq_hybrid} holds if and only if, for each integral curve $c\colon \mathbb{R}\rightarrow U_k$ of $X_{H}^{\gamma_k}$, the curve $\gamma_k\circ c\colon \mathbb{R}\rightarrow T^\ast U_k$ is an integral curve of $\restr{X_H}{T^\ast U_k}$. Since the dynamics $x(t)$ of $\hybrid_{H}$ are given by the integral curves of $X_H$ for $x(t)\notin S_H$, the curves $x(t)=\gamma_k \circ c(t)$ are the hybrid dynamics on $T^\ast U_k$.
    
    {On the other hand, if $c$ is a continuous curve that intersects $\pi_Q(S_H)$ at $\{t_k\}_k$ and $c(t_{k+1}) \in \partial U_k \cap \partial U_{k+1} \cap \partial \big(\pi_Q(S_H)\big)$, then $\gamma_{k+1}$ and $\gamma_k$ are $\Delta_H$-related if and only if
    \begin{equation}
        \gamma_{k+1} \circ c(t_{k+1}) = \Delta_H \circ \gamma_k \circ c(t_{k+1}) \, .
    \end{equation}
    }
\end{proof}

\begin{definition}\label{def:complete_sols}
{Consider a solution $({\gamma_{k}})_{\lambda}$ of the hybrid Hamilton--Jacobi problem where each one-form $\gamma_k$ depends on $n$ additional parameters $\lambda\in \mathbb{R}^n $, and suppose that $\left\{\Phi_k\colon U_k \times \R^n \to T^\ast U_k \right\}_k$, where $\Phi_{k}(q, \lambda)=({\gamma_{k}})_{\lambda}(q)$, is a family of local diffeomorphisms.} We say that $({\gamma_{k}})_{\lambda}(q)$ is \emph{complete solution of the Hamilton--Jacobi problem} for $\hybrid_{H}$ 
if, for each $\lambda \in \R^n$, there exists a $\mu \in \R^n$ such that $\Image (\Delta_H \circ ({\gamma_{k}})_{\lambda})\subseteq \Image  ({\gamma_{l}})_{\mu}$.

A hybrid Hamiltonian system $\hybrid_{H}= (T^{*}Q, {S_H}, X_{H}, \Delta_H)$ with a complete solution of the Hamilton--Jacobi problem $\{({\gamma_{k}})_{\lambda}(q)\}$ is called an \emph{integrable hybrid Hamiltonian system}.
\end{definition}

{Before the first impact takes place, the initial value of the parameters $\lambda\in \RR^n$ will be fixed by the initial conditions of the dynamics, and the subsequent values of these parameters will be determined by the relation $\operatorname{Im} (\Delta_H \circ ({\gamma_{k}})_{\lambda})\subseteq \operatorname{Im} ({\gamma_{l}})_{\mu}$.}



\begin{remark}\label{remark:Liouville_integrability_hybrid}
    Let $\mathfrak{U} = \left\{(U_k,\Phi_k)  \right\}$ be a complete solution of the hybrid Hamilton--Jacobi problem for $\hybrid_{H}= (T^{*}Q, {S_H}, X_{H}, \Delta_H)$. Clearly, for each $k$, $\Phi_k\colon U_k\times \R^n\to T^\ast U_k$ is a complete solution of the Hamilton--Jacobi problem for $\restr{H}{T^\ast U_k}$. Therefore, each of the open subsets $T^\ast U_k\subset T^\ast Q$ has a foliation on Lagrangian submanifolds invariant under the flow of $\restr{X_H}{T^\ast U_k}$ (see Remark~\ref{remark:Liouville_integrability}). These
    submanifolds are given by $\Im\Phi_k(\cdot, \lambda)$ for each $\lambda\in \R^n$.
    Moreover, the functions $f_{k,i}= \pi_i \circ \Phi_k^{-1}\colon T^\ast U_k\to \R$ are independent constants of the motion in involution. However, in general, they are not hybrid constants of the motion. {The condition $\operatorname{Im} (\Delta_H \circ ({\gamma_{k}})_{\lambda})\subseteq \operatorname{Im} ({\gamma_{l}})_{\mu}$ implies that the impact map $\Delta_H$ maps the border $\operatorname{Im}  ({\gamma_{l}})_{\lambda}\cap S_H$ of each Lagrangian submanifold of $T^\ast U_k$ into a subset of a Lagrangian submanifold $\operatorname{Im} ({\gamma_{l}})_{\mu}$ of $T^\ast U_l$. If $U_k$ is an open cover of $Q\setminus \pi_Q(S_H)$, then the Lagrangian submanifolds $\operatorname{Im} ({\gamma_{l}})_{\mu}$ will foliate the phase space $\cT Q$. The parameters $\mu$ and $\lambda$ are related by the condition $\operatorname{Im} (\Delta_H \circ ({\gamma_{k}})_{\lambda})\subseteq \operatorname{Im} ({\gamma_{l}})_{\mu}$.
    }
    Therefore, an integrable hybrid Hamiltonian system consists of integrable Hamiltonian systems related by the impact map.
    Furthermore, there exist action-angle coordinates in which the flow of $\restr{X_H}{T^\ast U_k}$ trivializes (see \cite{Arnold1978a, Audin2004, Bolsinov2004, Wiggins2003}). Hence, obtaining a complete solution of the Hamilton--Jacobi problem can be useful for computing the dynamics of certain hybrid Hamiltonian systems. The motion between one impact and the next one will be quasi-periodic (as a consequence of Liouville--Arnol'd theorem). {Nevertheless, an integrable hybrid Hamiltonian system may be chaotic. For instance, a billiard (i.e., a particle undergoing free motion within a plane, encountering collisions with a wall) is an integrable hybrid Hamiltonian system and there are chaotic billiards \cite{Chernov2006,K.T1991}}.

    { As a matter of fact, we have recently proven a Liouville--Arnol'd type theorem for hybrid Hamiltonian systems \cite{L.C2023}. 
    A \emph{Liouville integrable hybrid Hamiltonian system} is a $5$-tuple $(\cT Q, S_H, X_H, \Delta_H, F)$, formed by a hybrid Hamiltonian system $(\cT Q, S_H, X_H, \Delta_H)$, together with a map $F=(f_1, \ldots, f_n)\colon \cT Q \to \RR^n$ such that $\operatorname{rank} \T_x F = n$ almost everywhere and the functions $f_1, \ldots, f_n$ satisfy:
    \begin{enumerate}
        \item $X_H(f_i) = 0$ for all $i\in \{1, \ldots, n\}$,
        \item for each connected component $C\subseteq S_H$ and each $a_i\in \Image f_i$, there exists a $b_i\in \Image f_i$ such that
        \begin{equation}\label{eq:condition_generalized_hybrid_constant}
            \Delta\left(\restr{f_i}{C}^{-1}(a_i)\right)  \subseteq f_i^{-1}(b_i)\, ;
        \end{equation}
        \item $\{f_i, f_j\}=0$ for all $i, j\in \{1, \ldots, n\}$, where $\{\cdot, \cdot\}$ denotes the canonical Poisson bracket.
    \end{enumerate} 

    Note that the existence of a complete solution of the Hamilton--Jacobi problem for $\hybrid_H$ is a sufficient condition for it to be a Liouville integrable hybrid Hamiltonian system. 

    \begin{theorem}[Liouville--Arnol'd theorem for hybrid Hamiltonian systems]
        Consider a completely integrable hybrid Hamiltonian system $(\cT Q, S_H, X_H, \Delta_H)$, with $F=(f_1, \ldots, f_n)$, where $n=\dim Q$. Let $M_{\Lambda}$ be a regular level set of $F$, namely, $M_\Lambda= F^{-1}(\Lambda)$ for $\Lambda\in \RR^n$ such that $\operatorname{rank} \T_x F = n$ for all $x\in M_{\Lambda}$. Then:
        \begin{enumerate}
            \item Each regular level set $M_\Lambda$ is a Lagrangian submanifold of $\cT Q$, and it is invariant with respect to the flows of $X_H, X_{f_1}, \ldots, X_{f_n}$.
            \item Any compact connected component of $M_\Lambda$ is diffeomorphic to an $n$-dimensional torus $\mathbb{T}^n$.
            \item For each regular level set $M_\Lambda$ and each connected component $C\subseteq S$, there exists a $\Lambda'\in \RR^n$ such that $\Delta(M_\Lambda \cap C) \subset M_{\Lambda'} = F^{-1}(\Lambda')$. 
            \item On a neighbourhood $U_\lambda$ of $M_\Lambda$ there are coordinates $(\varphi^i, s_i)$ such that
            \begin{enumerate}
                \item they are Darboux coordinates for $\omega$, namely, $\omega = \dd \varphi^i \wedge \dd s_i$,
                \item the action coordinates $s_i$ are functions depending only on the functions $f_1, \ldots, f_n$,
                \item the continuous part hybrid dynamics are given by
                 $$\dot \varphi^i = \Omega^i(s_1, \ldots, s_n),\qquad \dot s_i = 0\, .
                 $$
                 \item In these coordinates, for each connected component $C\subseteq S$, the impact map reads $\Delta_H \colon (\varphi^i_-, s_i^-)\in M_\Lambda \cap C\mapsto  (\varphi^i_+, s_i^+)\in M_{\Lambda'}$, where $s_1^+, \ldots, s_n^+$ are functions depending only on $s_1^-, \ldots, s_n^-$.
            \end{enumerate}
        \end{enumerate}
    \end{theorem}
    }
\end{remark}

{It is worth mentioning that in \cite{Clark2020} (completely) integrable hybrid Hamiltonian systems are defined in terms of Lagrangian submanifolds. 
In that reference (see also \cite{Clark2021}), a submanifold $\mathcal{L}\subset T^\ast Q$ is called a hybrid Lagrangian submanifold if $\mathcal{L}\setminus \partial \mathcal{L}$ is a Lagrangian submanifold of $T^\ast Q$ and $\pi_Q(\partial \mathcal{L})$ is contained in $\pi_Q(\partial S_H)$ (see Definition V.III.6 in \cite{Clark2020}). 
More concretely (cf.~Definitions V.III.5 and V.III.6 in \cite{Clark2020}), a hybrid Hamiltonian system $\hybrid_{H}= (T^{*}Q, {S_H}, X_{H}, \Delta_H)$ is called completely integrable if there exists a foliation $\{\mathcal{L}_\alpha\}$ of $T^\ast Q$ such that, for each leaf $\mathcal{L}_\alpha$,
\begin{enumerate}
   \item \label{Clark_integrable_1_2} $\mathcal{L}_\alpha$ is a hybrid Lagrangian submanifold of $T^\ast Q$,
    \item \label{Clark_integrable_3} $\mathcal{L}_\alpha\subseteq H^{-1}(E)$ for some $E\in \RR$,
    \item \label{Clark_integrable_4} $\Delta_H(\partial \mathcal{L}_\alpha\cap S_H)\subseteq \mathcal{L}_\alpha$.
\end{enumerate}

{ Essentially, condition \ref{Clark_integrable_4} means that $\Delta_H$ maps points from the border of each Lagrangian leaf into points inside the leaf.}


{Let $I\subseteq \NN$ be a set of indices, $\{U_k\}_{k\in I}$ an open cover of $Q\setminus\pi_Q(S_H)$}, 
and $\left\{\Phi_k\colon U_k \times \R^n \to T^\ast U_k \right\}$ a complete solution of the Hamilton--Jacobi problem for $\hybrid_{H}$. Let $\alpha_k\in \RR^n$ such that $\Image \big(\Delta_H \circ \Phi_k(\cdot, \alpha_k)\big) \subseteq \Image \big( \Phi_l(\cdot, \alpha_l)\big)$ for each $k,l\in I$. 
{Definitions~\ref{def:solution_hybrid_HJ} and \ref{def:complete_sols} imply that $\mathcal{L}_\alpha = \cup_{k\in I} \Image \Phi_k(\cdot, \alpha_k)$ satisfies conditions \ref{Clark_integrable_1_2}. 
However, condition \ref{Clark_integrable_3} only holds if
\begin{equation}\label{eq:Clark_integrable}
    H\circ \Phi_k(\cdot, \alpha_k)=H\circ \Phi_l(\cdot, \alpha_l)=E\, ,
\end{equation}
for each $k\in I$. Additionally, condition \ref{Clark_integrable_4} holds if $\alpha_k = \alpha_l$.
}
In summary, an integrable hybrid Hamiltonian system (in the sense of Definition~\ref{def:complete_sols}) is completely integrable (in the sense of Clark) if and only if the open subsets $U_k$ cover $Q$, {$\alpha_k= \alpha_l$} and equation~\eqref{eq:Clark_integrable} holds. 
}

\begin{remark}\label{remark_reconstruction}
    Solutions of the Hamilton--Jacobi problem can be used to reconstruct the dynamics of a hybrid Hamiltonian system $\hybrid_{H}= (T^{*}Q, {S_H}, X_{H}, \Delta_H)$ as follows:
    \begin{enumerate}
        \item \label{remark_reconstruction_item_1} Solve the Hamilton--Jacobi equation \eqref{HJ_eq_hybrid} for each $(\gamma_k)_{\lambda_k}$.
        \item \label{remark_reconstruction_item_2} Label the open subsets $U_k$ so that the initial condition is in $U_0$ and between $k$-th and $(k+1)$-th impacts the system is in $U_k$. Let $t_k\in \R$ denote the instants of the impacts, i.e., $\lim_{t\to t_k^-} c_k(t) \in \pi_Q(S_H)$. 
        \item \label{remark_reconstruction_item_3} Use that $(\gamma_{k+1})_{\lambda_{k+1}}$ and $(\gamma_k)_{\lambda_k}$ are $\Delta_H$-related to compute $\lambda_{k+1}$ in terms of $\lambda_k$, where $\lambda_0$ is determined by the initial conditions.
        \item \label{remark_reconstruction_item_4} Compute the integral curves $c_k\colon (t_k, t_{k+1}) \to U_k$ of $X_H^{\gamma_k}$. 
        \item \label{remark_reconstruction_item_5} The hybrid dynamics are given by $\gamma_k\circ c_k(t)$ for each interval $(t_k, t_{k+1})$.
    \end{enumerate}
\end{remark}

\begin{example}[Circular billiard]\label{example:cicular_billiard}
Consider the hybrid Hamiltonian system $\hybrid_{H}= (T^{*}\RR^2 , {S_H}, X_{H}, \Delta_H)$, with Hamiltonian function $H= \frac{1}{2} \left(p_x^2+p_y^2\right)$, where the domain is $\X= \left\{(x,y, p_x, p_y) \in T^\ast \RR^2 \mid x^2+y^2\leq 1  \right\}$, the switching surface is $S_H = \left\{(x,y, p_x, p_y) \in T^\ast \RR^2 \mid x^2+y^2=1 { \hbox{ and } x p_x + y p_y >0}   \right\}$ and the impact map $\Delta\colon (x, y, p_x^-, p_y^-)\mapsto (x, y, p_x^+, p_y^+)$ is given by
\begin{equation}
    p_x^+ = p_x^- -2x\left(xp_x^- + y p_y^-\right),\,\,
    p_y^+ = p_y^- -2y\left(xp_x^- + y p_y^-\right).
\label{impact_map_billiard}
\end{equation}
They can be determined from the Newtonian law (see Remark \ref{remark:Newtonian_Hamiltonian_impact_law}) with the constraint function $h = 1-x^2-y^2$. 
Suppose that an impact occurs at $(x^\ast, y^\ast)$. Let $\gamma^-$ and $\gamma^+$ denote the solutions of the Hamilton--Jacobi problem before and after the impact, respectively, with $\gamma^\pm = \gamma_x^\pm \dd x + \gamma_y ^\pm \dd y$.  The Hamilton--Jacobi equation \eqref{HJ_eq_hybrid} is written 
\begin{align*}
    0 &= \frac{1}{2}  \dd \left(\left(\gamma_x^\pm\right)^2 + \left(\gamma_y^\pm\right)^2  \right)=\left(\gamma_x^\pm \frac{\partial \gamma_x^\pm} {\partial x} + \gamma_y^\pm \frac{\partial \gamma_y^\pm} {\partial  x}  \right) \dd x
    +\left(\gamma_x^\pm \frac{\partial \gamma_x^\pm} {\partial y} + \gamma_y^\pm \frac{\partial \gamma_y^\pm} {\partial  y}  \right) \dd y
\end{align*}
{Due to the form of this partial differential equation, it is natural to assume}
 that $\gamma$ is separable, i.e., $\gamma_x^\pm$ and $\gamma_y^\pm$ depend only on $x$ and on $y$, respectively.
Then,
\begin{align*}
    &\gamma_x^\pm \frac{\partial \gamma_x^\pm} {\partial x} = 0\quad \hbox{ and } \quad \gamma_y^\pm \frac{\partial \gamma_y^\pm} {\partial  y} = 0,
\end{align*}
so  $\gamma_x^\pm = a^\pm$ and $\gamma_y^\pm = b^\pm$, where $a^-, a^+, b^-, b^+$ are constants. {The parameters $a^-$ and $b^-$ determining the solution $\gamma^-$ before the impact are given by the initial values of the momenta, namely, $a^-=p_x(0)$ and $b^-= p_y(0)$. The parameters of the solution $\gamma^+$ after the impact are determined via the impact map \eqref{impact_map_billiard}:
  }
\begin{align*}
    a^+ = a^- -2x^\ast \left(x^\ast a^- + y^\ast  b^-\right)\hbox{ and }
    b^+ = b^- -2y^\ast \left(x^\ast a^- + y^\ast  b^-\right).
\end{align*}

One can check that the Hamiltonian function $H$ is a hybrid constant of the motion, i.e., $X_H(H)=0$ and $H\circ \Delta_H = \restr{H}{S}$. The angular momentum $\ell = xp_y - y p_x$ is another hybrid constant of the motion. Suppose that, along $\gamma^\pm$, $H$ and $\ell$ take the constant values $E$ and $\mu$, respectively. Then,
\begin{align*}
    & E = \frac{1}{2}\left( a^{-^2} + b^{-^2}\right)
    = \frac{1}{2}\left( a^{+^2} + b^{+^2}\right),\\
    & \mu =
    \left\{
    \begin{array}{ll}
        x(t) b^- -y(t) a^- & t< t^\ast,\\
        x(t) b^+ -y(t) a^+ & t> t^\ast,
    \end{array} \right.
\end{align*}
{
where $t^\ast$ denotes the time at which the impact occurs. Therefore, we can express $a^\pm$ and $b^\pm$ in terms of $E$ and $\mu$:
\begin{align*}
    & a^{-}=x(t)\sqrt{\frac{{2E}}{(x(t))^2+(y(t)+\mu)^2}}, \quad b^{-}=x(t)\frac{y(t)+\mu}{x(t)}\sqrt{\frac{{2E}}{(x(t))^2+(y(t)+\mu)^2}}, \quad t<t^\ast,\\
    & a^{+}=x(t)\sqrt{\frac{{2E}}{(x(t))^2+(y(t)+\mu)^2}}, \quad b^{+}=x(t)\frac{y(t)+\mu}{x(t)}\sqrt{\frac{{2E}}{(x(t))^2+(y(t)+\mu)^2}}, \quad t>t^\ast,\\
\end{align*}
}

{A complete solution of the Hamilton--Jacobi problem for $\hybrid_{H}$ is given by $\left\{(\gamma^-)_{(a^-, b^-)},(\gamma^+)_{(a^+, b^+)}\right\}$, where $(\gamma^\pm)_{(a^\pm, b^\pm)}=a^\pm \dd x + a^\pm \dd y$. Then, $f_1^\pm(x,y,p_x,p_y)=\pi_1 \circ (\Phi^{\pm})^{-1}(x,y,p_x,p_y)=p_x$ and $f_2^\pm(x,y,p_x,p_y)=\pi_2 \circ (\Phi^{\pm})^{-1}(x,y,p_x,p_y)=p_y$ are constants of the motion for $X_H$, but clearly they are not hybrid constants of the motion.}

{Observe that in this case the additional condition from Remark~\ref{remark:def_Clark} does not hold. Indeed, one can write $\gamma^\pm = \dd \mathcal{A}^\pm$, where
$\mathcal{A}^\pm = a^\pm x + b^\pm y$, but
$$\mathcal{A}^+(x^*, y^*)-\mathcal{A}^-(x^*, y^*)-=-2(x^*a^-+y^*b^-)\neq 0\, .$$
}

\end{example}

\begin{example}[Rolling disk hitting fixed walls]\label{example_disk}

Consider a homogeneous circular disk of radius $R$ and mass $m$ moving in the vertical plane $xOy$ (see {\cite[Example 8.2]{Ibort1997}}, and also {\cite[Example 3.7]{Ibort2001}}). Let $(x, y)$ be the coordinates of the center of the disk and $\vartheta$ the angle between a point of the disk and the axis $Oy$. The dynamics of the system is determined by the Hamiltonian $H$ on $T^{*}(\R^2 \times \mathbb{S}^1)$ given by
$H = \frac{1}{2m} \left(p_x^2 +  p_y^2\right) +\frac{1}{2mk^2} p_\vartheta^2$, which is hyperregular since it is mechanical.
{Here $m$ is the mass of the disk and $k$ is a constant such that $mk^2$ is the moment of inertia of the disk.}

Suppose that there are two rough walls at the axis $y=0$ and at $y=h$, where $h=\alpha R$ for some constant $\alpha>1$. Assume that the impact with a wall is such that the disk rolls without sliding and that the change of the velocity along the $y$-direction is characterized by an elastic constant $e$. 
{The condition of rolling without sliding is given by $p_x=R p_{\vartheta}/k^2$, whereas the change of momenta in the $y$-direction can be derived from the Newtonian impact law (see Remark \ref{remark:Newtonian_Hamiltonian_impact_law}) with the constraint function equal to $y$.} {Note that the condition of rolling without sliding is a nonholonomic constraint which is present only on the impact surface. The Hamilton--Jacobi theory for hybrid systems with nonholomic constraints on all the domain is studied in Section~\ref{sec:nonholonomic}.}
When the disk hits one of the walls, the impact map is given by (see {\cite[Example 8.2]{Ibort1997}}, and also {\cite[Example 3.7]{Ibort2001}})
\begin{equation}
    {
  \left(p_x^-, p_y^-, p_\vartheta^-  \right)
  \mapsto \left(\frac{R^2 p_x^- + R p_\vartheta^-}{k^2+R^2}, - ep_y^-, k^2 \frac{R p_x^- + p_\vartheta^-}{k^2+R^2}  \right),}
\end{equation}
where the switching surface is given by
{ $S_H = C_1 \cup C_2$, with
\begin{align}
    & C_1 = \{(x,y,\vartheta,p_x,p_y,p_\vartheta)\mid y=R,\,  p_x=R p_{\vartheta}/k^2 \hbox{ and } p_y<0\}\, ,\\
    & C_2 = \{(x,y,\vartheta,p_x,p_y,p_\vartheta)\mid y=h-R,\,  p_x=R p_{\vartheta}/k^2 \hbox{ and } p_y>0\}\, .
\end{align}
Here the condition $p_x=R p_{\vartheta}/k^2$ comes from the nonholonomic constraint of the walls, whereas the conditions on the sign of $p_y$ ensure that the $y$-component of the momenta (or the velocity) points towards corresponding the wall.
}

We look for a separable solution of the hybrid Hamilton--Jacobi problem, i.e., a collection of one-forms $\gamma_i=\gamma_{x,i}(x) \dd x + \gamma_{y,i}(y) \dd y + \gamma_{\vartheta,i}(\vartheta) \dd \vartheta$. From equation~\eqref{HJ_eq_hybrid}, we obtain
$\gamma_i=a_i \dd x + b_i \dd y + c_i \dd \vartheta$, where $a_i, b_i, c_i$ are constants. The relation between these constants is given by equation~\eqref{Delta_related_condition}:
\begin{equation}
    {
    a_{i+1} = \frac{R^2 a_i + R c_i}{k^2+R^2}\,, \quad b_{i+1} = -e b_i\,,\hbox{ and }c_{i+1} = k^2\frac{R a_i + c_i}{k^2+R^2}\, .}
\end{equation}
We conclude that $\{(\gamma_k)_{(a_k,b_k,c_k)}\}$ is a complete solution of the Hamilton--Jacobi problem for $\hybrid_{H}$. 
{The initial values $(a_0, b_0, c_0)$ correspond with the initial values $(p_x(0), p_y(0), p_{\vartheta}(0))$ of the momenta at time zero.
Each one-form $\gamma_i$ determines a Lagrangian submanifold of $\cT (\RR^2\times \mathbb{S}^1)$, namely,
\begin{equation}
    \mathcal{L}_i = \operatorname{Im} \gamma_i= \left\{ (x, y, \vartheta, p_x, p_y, p_\vartheta)\in \cT (\RR^2\times \mathbb{S}^1) \mid p_x = a_i,\, p_y = b_i,\, p_\vartheta = c_i\right\}\, .
\end{equation}
}
\end{example}


{
It is worth remarking that in the examples above the separability condition has been imposed since it is a natural \textit{ansatz} for Hamiltonians of the form $H = \sum_{i=1}^n a_i p_i^2$ for some constants $a_i$. However, this is by no means a requirement of our theory.
}

{

It is worth remarking that the addition of impacts to a continuous integrable system does not necessarily lead to an integrable hybrid system. As the following example illustrates, condition~\eqref{Delta_related_condition} may not be satisfied.

\begin{example}[Hybrid integrability is not guaranteed by continuous integrability]

Harmonic oscillators are a paradigmatic example of integrable system. Nevertheless, it is possible to consider a harmonic oscillator with impacts which is not an integrable hybrid system. More precisely, consider the hybrid Hamiltonian system $\hybrid_H=(\cT \RR^2, S_H, X_H, \Delta_H)$, where 
the Hamiltonian function $H\colon \cT \RR^2 \to \RR$ given by
\begin{equation}
    H = \frac{1}{2}\left(p_x^2 + p_y^2 + x^2 +y^2\right)\, ,
\end{equation}
in canonical bundle coordinates, while the switching surface $S_H$ and the impact map $\Delta_H$ are as in Example~\ref{example:cicular_billiard}. A complete solution of the (continuous) Hamilton--Jacobi problem for $H$ is given by 
\begin{equation}
    \Phi\colon (x,y, \lambda_x, \lambda_y) \in \RR^2\times \RR^2\mapsto \pm \sqrt{\lambda_x - x^2} \dd x \pm \sqrt{\lambda_y - y^2} \dd y \in \cT \RR^2\, . 
\end{equation}
Consequently,
\begin{equation}
\begin{aligned}
    \Delta_H \circ \Phi (x,y, \lambda_x, \lambda_y) 
    & = \left(\pm (1-2x^2) \sqrt{\lambda_x - x^2} \mp 2xy \sqrt{\lambda_y-y^2}\right) \dd x 
    \\ &\quad 
    + \left(\pm (1-2y^2) \sqrt{\lambda_y - y^2} \mp 2xy \sqrt{\lambda_x-x^2}\right) \dd y\, ,
\end{aligned}
\end{equation}
and there exists no $(\mu_x, \mu_y)\in \RR^2$ such that
$ \Delta_H \circ \Phi (x,y, \lambda_x, \lambda_y) = \Phi(x, y, \mu_x, \mu_y)$ for all $(x, y)\in S_H$.
\end{example}
}



\section{Geometric Hamilton--Jacobi theory for forced hybrid systems}\label{sec4}

\begin{definition}
Consider a forced hybrid Hamiltonian system $\hybrid_{H, F}= (T^{*}Q, {S_H}, X_{H,F}, \Delta_H)$. Let $U_k\subset Q\setminus \pi_Q(S_H)$ be open subsets of $Q$. Let $\mathfrak{U} = \left\{(U_k, \gamma_k)  \right\}$ be a family of closed one-forms $\gamma_k$ on $U_k$. We will say that $\mathfrak{U}$ {is a solution of} the \emph{hybrid Hamilton--Jacobi} equation for $\hybrid_{H, F}$ if:
\begin{enumerate}
\item For each $k$, $\gamma_k$ is a solution of the forced Hamilton--Jacobi equation on $U_k$, namely,
\begin{equation}
    \dd (H \circ \gamma_k) = -\gamma_k^\ast F. 
    \label{HJ_eq_hybrid_forc}
\end{equation}
\item If $x\in \partial U_k \cap \pi_Q(S_H)$, then there exists an $l$ such that $x\in \partial U_l \cap \pi_Q(S_H)$ and $\gamma_l$ and $\gamma_k$ are $\Delta_H$-related, i.e.,
\begin{equation}
    \lim_{y\to x} \gamma_l(y) = \Delta_H \left( \lim_{y\to x} \gamma_k(y) \right).
\label{Delta_related_condition_forc}
\end{equation}

\end{enumerate}

\end{definition}

\begin{theorem}[Forced hybrid Hamilton--Jacobi theorem]\label{HJ_theorem_hybrid_forced}
Let $\hybrid_{H,F}= (T^{*}Q, {S_H}, X_{H,F}, \Delta_H)$ be a forced hybrid Hamiltonian system, $U_k\subset Q\setminus \pi_Q(S_H)$ be open subsets of $Q$ and $\mathfrak{U} = \left\{(U_k, \gamma_k)  \right\}$ be a family of closed one-forms $\gamma_k\in \Omega^1(U_k)$.  Then, the following statements are equivalent:
\begin{enumerate}
    \item The family $\mathfrak{U}$ {is a solution of} the hybrid Hamilton--Jacobi equation for $\hybrid_{H,F}$. 
{
    \item For every curve $c\colon \R\to Q$ such that 
    \begin{enumerate}
        \item $c\big((t_k, t_{k+1})\big)\subset U_k$, 
        \item $c$ intersects $\pi_Q(S_H)$ at $\{t_k\}_k$,
        \item $c$ satisfies the equations
    \begin{align}
    &\dot c(t) = T \pi_Q \circ  X_{H,F} \circ \gamma_k \circ c(t), \qquad t_k<t<t_{k+1},\\
    & \gamma_{k+1} \circ c(t_{k+1}) = \Delta_H \circ \gamma_k \circ c(t_{k+1}),
    \end{align}
    \end{enumerate}
    the curve $\tilde{c}\colon \RR \to T^\ast Q$ given by $\tilde{c}(t) = \gamma_k \circ c(t)$ for $t\in[t_k, t_{k+1})$ is an integral curve of the hybrid dynamics.}
\end{enumerate}
\end{theorem}

\begin{proof}
    The proof follows from the one of Theorem \ref{HJ_theorem_hybrid} by replacing the Hamilton--Jacobi equation \eqref{HJ_eq_hybrid} and the Hamiltonian vector $X_H$ field by the forced Hamilton--Jacobi equation \eqref{HJ_eq_hybrid_forc} and the forced Hamiltonian vector field $X_{H,F}$, respectively.
\end{proof}


Complete solutions are defined as in the unforced case (see Definitions \ref{def_complete_sols} and \ref{def:complete_sols}). The dynamics of $\hybrid_{H, F}$ can be reconstructed by solving the forced Hamilton--Jacobi equation \eqref{HJ_eq_hybrid_forc} for each $(\gamma_k)_{\lambda_k}$ and following {steps} \ref{remark_reconstruction_item_2}-\ref{remark_reconstruction_item_5} from Remark \ref{remark_reconstruction}. (Observe that $T\pi_Q(X_{H,F})=T\pi_Q(X_H)$, so step \ref{remark_reconstruction_item_4} does not change.)


\begin{remark}\label{remark:Liouville_integrability_hybrid_forced}
    Given a complete solution $\left\{\Phi_k\colon U_k \times \R^n \to T^\ast U_k \right\}$, each $T^\ast U_k$ can be foliated into Lagrangian submanifolds invariant under the flow of $\restr{X_{H,F}}{T^\ast U_k}$ (see Remarks~\ref{remark:Liouville_integrability_forced} and \ref{remark:Liouville_integrability_hybrid}). {Moreover, the impact map $\Delta_H$ maps the borders of the Lagrangian submanifolds of $T^\ast U_k$ into subsets of the Lagrangian submanifolds of $T^\ast U_l$.} In addition, the functions $f_{k,i}= \pi_i~\circ~\Phi_k^{-1}\colon T^\ast U_k\to \R$ are constants of the motion for $\left(\restr{H}{\cT U_k}, \restr{F}{\cT U_k}\right)$ but, in general, they are not hybrid constants of the motion. 
    {Therefore, obtaining a complete solution of the Hamilton--Jacobi problem may simplify solving the equations of motion for $\hybrid_{H, F}$. For instance, one may choose coordinates $(x_i, y_i)$ in $U_k$, where $x_i = f_{k,i}$, so that half of the $2n$ equations of motion on $U_k$ become simply $\restr{X_{H,F}}{U_k} (x_i) = 0$.}
\end{remark}

\begin{example}[Rolling disk with dissipation hitting fixed walls]\label{example_disk_forced}

Consider the hybrid system from Example \ref{example_disk}. Suppose that now the disk is additionally subject to the external force
\begin{equation}
     F = B p_y \dd y,
\end{equation}
where $B$ is some constant. 
Then, a one-form $\gamma_k = \gamma_{k,x} \dd x + \gamma_{k,y} \dd y + \gamma_{k,\vartheta} \dd \vartheta$ satisfies equation~\eqref{HJ_eq_hybrid_forc} if and only if
\begin{align}
    &\frac{1}{m} \left[\gamma_{k,x} \frac{\partial \gamma_{k,x}}{\partial x}
    + \gamma_{k,y} \frac{\partial \gamma_{k,y}}{\partial x}
    + \frac{1}{k^2} \gamma_{k,\vartheta} \frac{\partial \gamma_{k,\vartheta}}{\partial x}
    \right] = 0,\\ 
    &\frac{1}{m} \left[\gamma_{k,x} \frac{\partial \gamma_{k,x}}{\partial y}
    + \gamma_{k,y} \frac{\partial \gamma_{k,y}}{\partial y}
    + \frac{1}{k^2} \gamma_{k,\vartheta} \frac{\partial \gamma_{k,\vartheta}}{\partial y}
    \right] = B \gamma_{k,y},\\
    &\frac{1}{m} \left[\gamma_{k,x} \frac{\partial \gamma_{k,x}}{\partial \vartheta}
    + \gamma_{k,y} \frac{\partial \gamma_{k,y}}{\partial \vartheta}
    + \frac{1}{k^2} \gamma_{k,\vartheta} \frac{\partial \gamma_{k,\vartheta}}{\partial \vartheta}
    \right] = 0,\\ 
\end{align}
Suppose that $\gamma_k$ is separable, namely, $\gamma_{k,x}=\gamma_{k,x}(x), \gamma_{k,y}=\gamma_{k,y}(y), \gamma_{k,\vartheta}=\gamma_{k,\vartheta}(\vartheta)$. Then,
\begin{align}
    &\gamma_{k,x} \frac{\dd \gamma_{k,x}}{\dd x} = 0
    \, \gamma_{k,y} \frac{\dd \gamma_{k,y}}{\dd y} = Bm \gamma_{k,y},\hbox{ and }\frac{1}{k^2} \gamma_{k,\vartheta} \frac{\dd \gamma_{k,\vartheta}}{\dd \vartheta} = 0,
\end{align}
so
\begin{equation}
    \gamma_k = a_k\ \dd x + (Bm y + b_k)\ \dd y + c_k\ \dd \vartheta.
\end{equation}
The relation between $(a_k, b_k, c_k)$ and $(a_l, b_l, c_l)$ is given by equation~\eqref{Delta_related_condition_forc}:
\begin{align}
    & a_l = \frac{R^2 a_k   + k^2 R c_k}{k^2+R^2},\\
    & Bmy + b_l = -e(Bmy+b_k),\\
    & c_l = \frac{R a_k   + k^2 c_k}{k^2+R^2}.
\end{align}
Thus,
\begin{equation}
    b_l 
    = \left\{ 
    \begin{array}{ll}
        -eb_k   & \text{for the wall at } y=0,  \\
        -(e+1)Bmh-eb_k   & \text{for the wall at } y=h.  \\
    \end{array}
    \right.
\end{equation}
We conclude that $\{(\gamma_k)_{(a_k,b_k,c_k)}\}$ is a complete solution of the Hamilton--Jacobi problem for $\hybrid_{H, F}$.

\end{example}

\section{Geometric Hamilton--Jacobi theory nonholonomic hybrid systems}\label{sec:nonholonomic}
\subsection{Nonholonomic systems}
Suppose that $(Q,g)$ is a Riemannian manifold. Let $L\colon TQ \to \R$ be a mechanical Lagrangian, namely,
\begin{equation}
    L(q,v) = \frac{1}{2} g_q(v,v) - V(q).
\end{equation}
Then, the Legendre transform is simply the flat isomorphism defined by the metric $g$, namely,
\begin{equation}
\begin{aligned}
    \mathbb{F} L = \flat_g \colon TQ &\to T^\ast Q\\
    \left(q^i, \dot q^i \right) &\mapsto \left(q^i, g_{ij} \dot q^j \right).
\end{aligned}
\end{equation}
Let $H\colon T^\ast Q \to \RR$, where 
\begin{equation}
    H(q,p) = \frac{1}{2} g_q^{-1}(p,p) + V(q),
\end{equation}
be the Hamiltonian function corresponding to $L$, {with $g_q^{-1}$ denoting the inverse of $g$ at $q$.}

Suppose that the system is subject to the (linear) nonholonomic constraints given by the distribution
\begin{equation}
    \mathcal D = \left\{ v \in TQ \mid \mu^a(v) = 0,\ a=1,\ldots, k \right \},
\end{equation}
{where $\mu^a=\mu^a_i(q) \dd q^i$ are constraint one-forms.}
Denote by $\mathcal M = \mathbb{F}L (\mathcal D)$ the associated codistribution. Then, the \emph{nonholonomic vector field} $X_H^{\mathrm{nh}}$ for $H$ is given by
\begin{equation}
    \omega_Q(X_H^{\mathrm{nh}}, \cdot) = \dd H - \lambda_a \ \mu^a,
\end{equation}
with the constraint 
\begin{equation}
    T\pi_Q\left(X_H^{\mathrm{nh}}\right) \in \mathcal{D}.
\end{equation}
Here, $\lambda_a$ are Lagrange multipliers. Next, let $Z_a$ be the vector fields given by $\omega_Q(Z_a, \cdot) = \mu^a$, so that $X_H^{\mathrm{nh}} = X_H - \lambda_a Z_a$.
Locally, $Z_a = \mu^a_i \frac{\partial}{\partial p_i}$, so the projections on $Q$ of the nonholonomic and the Hamiltonian vector fields coincide, i.e., $T\pi_Q\left(X_H^{\mathrm{nh}}\right)= T\pi_Q\left(X_H\right)$.

Hereinafter, assume that $\mathcal{D}$ is a \textit{completely nonholonomic distribution}, that is,
\begin{equation}
    TQ = \mathrm{span} \left\{\mathcal D, [\mathcal D, \mathcal D], [\mathcal D, [\mathcal D, \mathcal D]], \ldots \right\}.
\end{equation}

\begin{theorem}[Ohsawa and Bloch, 2009]\label{HJ_theorem_nh}
Let $\gamma$ be a one-form on $Q$ such that $\operatorname{Im} \gamma \subset \mathcal M$ and
    $\dd \gamma(v,w)=0$ for any $v,w\in \mathcal D$. Then, the following statements are equivalent:
    \begin{enumerate}[label=\roman*)]
        \item For every integral curve $c$ of $T\pi_Q \circ X_H \circ \gamma$, the curve $\gamma \circ c$ is an integral curve of $X_H^{\mathrm{nh}}$.
        \item The one-form $\gamma$ satisfies the nonholonomic Hamilton--Jacobi equation:
        \begin{equation}
            H \circ \gamma = E,
        \end{equation}
        where $E$ is a constant.
    \end{enumerate}
\end{theorem}
See \cite{Ohsawa2009} for the proof.

\begin{remark}
    In the general case, the Hamilton--Jacobi equation is
    {
    \begin{equation}
       \dd (H \circ \gamma) \in \mathcal D^\circ,
    \end{equation}
    where $\mathcal D^\circ$ denotes the annihilator of $ \mathcal D$,
    }
    since it is not possible to apply Chow's theorem (see \cite{Iglesias-Ponte2008}). For an approach based on Lie algebroids, see \cite{deLeon2010a}.
\end{remark}

\begin{remark}
    In \cite{deLeon2014}, the authors considered the Hamilton--Jacobi equation for a Hamiltonian system on an almost-Poisson manifold $(E, \Lambda)$ fibered over a base manifold $Q$. Moreover, they defined complete solutions for the Hamilton--Jacobi equation, and showed that from a complete solution one can define {$m=\mathrm{rank} E$} independent functions in involution with respect to the almost-Poisson bracket (see Definition~4.1 and Proposition~2 in \cite{deLeon2014}, respectively). In particular, in Subsection 5.2 of that paper, they consider the nonholonomic almost-Poisson bracket. Therefore, a complete solution of the Hamilton--Jacobi equation for a nonholonomic system provides {$m=\mathrm{rank}\,\mathcal M$} independent functions in involution with respect to the nonholonomic bracket. However, since almost-Poisson brackets do not necessarily verify the Jacobi identity, these functions in involution are not associated with commuting vector fields as in the Poisson case. Nevertheless, we have recently shown an isomorphism between the nonholonomic bracket and the so-called Eden bracket (see~\cite{deLeon2023, deLeon2023a}), namely,
    $$ \{f, g\}_{nh} = \{f \circ \gamma , g \circ \gamma\} \circ i_{\mathcal{M}}\, ,$$
    for $f, g\in C^\infty(\mathcal{M})$, where $\gamma\colon T^\ast Q \to \mathcal{M}$ is an orthogonal projector {with respect to the metric appearing in the mechanical Hamiltonian,} and $i_\mathcal{M}\colon \mathcal{M} \to T^\ast Q$ is the canonical inclusion. Therefore, roughly speaking, nonholonomic dynamics can be understood as free Hamiltonian dynamics for a modified Hamiltonian function restricting to the appropriate subspaces. Consequently, complete solutions of the Hamilton--Jacobi problem for a nonholonomic system seem to be related to complete integrability in the sense of Liouville. We will explore this relation in future works.
\end{remark}

\subsection{Nonholonomic hybrid systems}

A  hybrid dynamical system is said to be a \emph{nonholonomic hybrid system} if it is determined by $\hybrid_{\mathrm{nh}}\coloneqq (T^\ast Q, \mathcal M, {S_H}, X_{H}^{\mathrm{nh}}, \Delta_H)$, where $X_{H}^{\mathrm{nh}}$ is the nonholonomic vector field associated with the Hamiltonian function $H$ with constraint codistribution $\mathcal{M}$, ${S_H}$ is the switching surface, a submanifold of $T^\ast Q$ with co-dimension one, and $\Delta_H:{S}_H\to \mathcal M$ is the impact map, a smooth embedding.

\begin{definition}
Consider a nonholonomic hybrid system $\hybrid_{\mathrm{nh}}= (T^\ast Q, \mathcal M, {S_H}, X_{H}^{\mathrm{nh}}, \Delta_H)$. Let $U_k\subset Q\setminus \pi_Q(S_H)$ be open subsets of $Q$. Let $\mathfrak{U} = \left\{(U_k, \gamma_k)  \right\}$ be a family of one-forms $\gamma_k$ on $U_k$ such that $\Im\gamma_k\subset \mathcal M$ and $\dd \gamma_k(v,w)=0$ for each $v, w \in \mathcal D$.
We will say that $\mathfrak{U}$ {is a solution of} the \emph{hybrid Hamilton--Jacobi} equation for $\hybrid_{\mathrm{nh}}$ if:
\begin{enumerate}
\item For each $k$, $\gamma_k$ is a solution of the hybrid Hamilton--Jacobi equation on $U_k$, namely,
\begin{equation}
     H \circ \gamma_k = E_k. 
    \label{HJ_eq_hybrid_nh}
\end{equation}
\item If $x\in \partial U_k \cap \pi_{Q}(S_H)$, then there exists an $l$ such that $x\in \partial U_l \cap \pi_{Q}(S_H)$ and $\gamma_l$ and $\gamma_k$ are $\Delta_H$-related, i.e.,
\begin{equation}
    \lim_{y\to x} \gamma_l(y) = \Delta_H \left( \lim_{y\to x} \gamma_k(y) \right).
\label{Delta_related_condition_nh}
\end{equation}
\end{enumerate}
\end{definition}

\begin{theorem}[Nonholonomic hybrid Hamilton--Jacobi theorem]\label{HJ_theorem_hybrid_nh}
Let $\hybrid_{\mathrm{nh}}= (T^\ast Q, \mathcal M, {S_H}, X_{H}^{\mathrm{nh}}, \Delta_H)$ be a nonholonomic hybrid system, $U_k\subset Q\setminus \pi_Q(S_H)$ be open subsets of $Q$ and $\mathfrak{U} = \left\{(U_k, \gamma_k)  \right\}$ be a family of one-forms $\gamma_k\in \Omega^1(U_k)$ such that $\Im\gamma_k\subset \mathcal M$ and $\dd \gamma_k(v,w)=0$ for each $v, w \in \mathcal D$. Then, the following statements are equivalent:
\begin{enumerate}
    \item The family $\mathfrak{U}$ {is a solution of} the hybrid Hamilton--Jacobi equation for $\hybrid_{\mathrm{nh}}$. 
{
    \item For every curve $c\colon \R\to Q$ such that 
    \begin{enumerate}
        \item $c\big((t_k, t_{k+1})\big)\subset U_k$, 
        \item $c$ intersects $\pi_Q(S_H)$ at $\{t_k\}_k$,
        \item $c$ satisfies the equations
    \begin{align}
    &\dot c(t) = T \pi_Q \circ  X_H^{\mathrm{nh}} \circ \gamma_k \circ c(t), \qquad t_k<t<t_{k+1},\\
    & \gamma_{k+1} \circ c(t_{k+1}) = \Delta_H \circ \gamma_k \circ c(t_{k+1}),
    \end{align}
    \end{enumerate}
    the curve $\tilde{c}\colon \RR \to\mathcal{M}$ given by $\tilde{c}(t) = \gamma_k \circ c(t)$ for $t\in[t_k, t_{k+1})$ is an integral curve of the hybrid dynamics.}
\end{enumerate}
\end{theorem}

\begin{proof}
    By Theorem \ref{HJ_theorem_nh}, equation~\eqref{HJ_eq_hybrid_nh} holds if and only if, for each integral curve $c\colon \mathbb{R}\rightarrow U_k$ of $X_{H}^{\gamma_k}$, the curve $\gamma_k\circ c\colon \mathbb{R}\rightarrow T^\ast U_k$ is an integral curve of $\restr{X^{\mathrm{nh}}_{H}}{T^\ast U_k}$. Since the dynamics $x(t)$ of $\hybrid_{\mathrm{nh}}$ are given by the integral curves of $X_H^{\mathrm{nh}}$ for $x(t)\notin S_H$, the curves $x(t)=\gamma_k \circ c(t)$ are the hybrid dynamics on $T^\ast U_k\cap \mathcal M$.
    
    {
    On the other hand, if $c$ is a continuous curve that intersects $\pi_Q(S_H)$ at $\{t_k\}_k$ and $c(t_{k+1}) \in \partial U_k \cap \partial U_l \cap \partial \big(\pi_{Q}(S_H)\big)$, then $\gamma_{k+1}$ and $\gamma_k$ are $\Delta_H$-related if and only if
    \begin{equation}
        \gamma_{k+1} \circ c(t_{k+1}) = \Delta_H \circ \gamma_k \circ c(t_{k+1}) \, .
    \end{equation}
    }
\end{proof}


\begin{definition}\label{def:complete_sols_nh}
Let $m=\operatorname{rank} \mathcal D=\operatorname{rank} \mathcal M$. 
Consider a family of solutions $({\gamma_{k}})_{\lambda}$ depending on $m$ additional parameters $\lambda\subset \mathbb{R}^{m}$, and suppose that $\left\{\Phi_k\colon U_k \times \R^{m} \to T^\ast U_k\cap \mathcal M \right\}$, where $\{\Phi_{k}(q, \lambda)=({\gamma_{k}})_{\lambda}(q)\}$, is a family of local diffeomorphisms. We say that $({\gamma_{k}})_{\lambda}(q)$ is \emph{complete solution of the Hamilton--Jacobi problem} for $\hybrid_{\mathrm{nh}}$ if, for each $\lambda \in \R^{m}$, there exists a $\mu \in \R^{m}$ such that $\operatorname{Im} (\Delta_H \circ ({\gamma_{k}})_{\lambda})\subseteq \operatorname{Im}  ({\gamma_{l}})_{\mu}$.

A nonholonomic hybrid Hamiltonian system $\hybrid_H= (T^{*}Q, {S_H}, X_{H}, \Delta_H)$ with a complete solution of the Hamilton--Jacobi problem $\{({\gamma_{k}})_{\lambda}(q)\}$ is called an \emph{integrable nonholonomic hybrid Hamiltonian system}.
\end{definition}



Solutions of the Hamilton--Jacobi problem can be used to reconstruct the dynamics of a nonholonomic hybrid Hamiltonian system (see Remark~\ref{remark_reconstruction}).


\subsection{Examples}
\subsubsection{The nonholonomic particle.}
    Consider a mechanical system on $Q= \R^3$, with Lagrangian function
    \begin{equation}
        L = \frac{1}{2} \left(\dot x ^2 + \dot y^2 + \dot z^2\right),
    \end{equation}
    and constraint distribution
    \begin{equation}
        \mathcal D = \mathrm{span} \left\{ \frac{\partial}{\partial x} + y \frac{\partial}{\partial z}, \frac{\partial}{\partial y} \right\}. 
    \end{equation}
    Then, the Hamiltonian function is
    \begin{equation}
        H = \frac{1}{2} \left(p_x ^2 + p_y^2 + p_z^2\right),
    \end{equation}
    and the constraint codistribution is
    \begin{equation}
        \mathcal M = \mathrm{span} \left\{ \dd x + y \dd z, \dd y\right\}.
    \end{equation}
    Let $\gamma = \gamma_x \dd x + \gamma_y \dd y + \gamma_ z \dd z$ be a solution of the nonholonomic Hamilton--Jacobi equation. The condition $\operatorname{Im}\gamma \subset \mathcal M$ means that $\gamma$ is a linear combination of $\dd x + y \dd z$ and $\dd y$, which implies that $\gamma_z=y \gamma_x$. The Hamilton--Jacobi equation can the be written as
    \begin{equation}
        E = \frac{1}{2} \left(\gamma_x^2 + \gamma_y ^2 + \gamma_z ^2\right)
         = \frac{1}{2} \left(\gamma_x^2 + \gamma_y ^2 + y^2 \gamma_x^2\right),
    \end{equation}
    so
    \begin{equation}
        \gamma_y = \pm \sqrt{2E-(1+y^2) \gamma_x^2},
    \end{equation}
    and then
    \begin{equation}
        \gamma = \gamma_x \dd x \pm \sqrt{2E-(1+y^2) \gamma_x^2} \dd y + y \gamma_x \dd z.
    \end{equation}
    We finally have to impose the condition $\dd \gamma(v,w)$ for any $v, w \in \mathcal D$. Since $\dd \gamma$ is bilinear and alternating, we just have to impose that 
    \begin{equation}
        \dd \gamma \left(\frac{\partial}{\partial y}, \frac{\partial}{\partial x} + y \frac{\partial}{\partial z}\right) = 0.
    \end{equation}
    Let us assume that $\gamma_x$ only depends on $y$. Then,
    \begin{equation}
        \dd \gamma \left(\frac{\partial}{\partial y}, \cdot \right)
        = \frac{\dd \gamma_x}{\dd y} \dd x + \left( \gamma_x + y \frac{\dd \gamma_x}{\dd y}\right) \dd z,
    \end{equation}
    so
    \begin{equation}
         \dd \gamma \left(\frac{\partial}{\partial y}, \frac{\partial}{\partial x} + y \frac{\partial}{\partial z}\right) 
         = (1+y^2) \frac{\dd \gamma_x}{\dd y} + y \gamma_x.
    \end{equation}
    Therefore,
    \begin{equation}
        \gamma_x = \frac{\lambda}{\sqrt{1+y^2}},
    \end{equation}
    for some $\lambda\in \RR$, and then
    \begin{equation}
        \gamma =  \frac{\lambda}{\sqrt{1+y^2}} \dd x 
        \pm \sqrt{2E- \lambda^2} \dd y
        + \frac{\lambda y}{\sqrt{1+y^2}} \dd z.
    \end{equation}

    Now suppose that there are two walls at the planes $y=0$ and $y=a$ and that the particle impacts with them with an elastic constant $e$, that is,
    \begin{equation}
    \Delta_H (p_x, p_y, p_z) = (p_x, -ep_y, p_z).
    \end{equation}
    {This impact map can be derived from the Newtonian impact law (see Remark \ref{remark:Newtonian_Hamiltonian_impact_law}) for a constraint function $h=y$.}
    Let $\gamma_i$ denote the solution of the Hamilton--Jacobi equation between the $i$-th and $(i+1)$-th impacts, namely,
    \begin{equation}
        \gamma_i =  \frac{\lambda_i}{\sqrt{1+y^2}} \dd x 
        \pm \sqrt{2E_i- \lambda_i^2} \dd y
        + \frac{\lambda_i y}{\sqrt{1+y^2}} \dd z.
    \end{equation}
    Then, the condition $\gamma_{i+1}=\Delta \circ \gamma_i$ implies that
    \begin{equation}
    \begin{aligned}
        \frac{\lambda_{i+1}}{\sqrt{1+y^2}} &= \frac{\lambda_i}{\sqrt{1+y^2}},\\
         \sqrt{2E_{i+1}- \lambda_{i+1}^2} &= e\sqrt{2E_i- \lambda_i^2},\\
        \frac{\lambda_{i+1} y}{\sqrt{1+y^2}} &= \frac{\lambda_i y}{\sqrt{1+y^2}},
    \end{aligned}
    \end{equation}
    that is,
    \begin{equation}
    \begin{aligned}
        & \lambda_{i+1} = \lambda_i,\\
        & E_{i+1} = e^2 E_i + \frac{1+e^2}{2} \lambda_{i+1}^2.
    \end{aligned}
    \end{equation}
    Clearly, fixing $\lambda_i$ and $E_i$ also fixes $\lambda_{i+1}$ and $E_{i+1}$, that is,
    \begin{equation}
        \operatorname{Im} (\Delta_H \circ ({\gamma_{i}})_{(\lambda_i, E_i)})\subseteq \operatorname{Im}  ({\gamma_{i+1}})_{(\lambda_{i+1}, E_{i+1})}.
    \end{equation}
    Moreover, $(x, y, z, \lambda_i, E_i)\mapsto \gamma_i$ is a local diffeomorphism, so $(\gamma_i)_{(\lambda_i, E_i)}$ is a complete solution.

    The Hamiltonian vector field of $H$ is $X_H = p_x \frac{\partial}{\partial x} + p_y \frac{\partial}{\partial y}
        + p_z \frac{\partial}{\partial z}$, so
    \begin{equation}
        X_H^{\gamma_i} = \gamma_{i,x} \frac{\partial}{\partial x} + \gamma_{i,y} \frac{\partial}{\partial y} + \gamma_{i,z} \frac{\partial}{\partial z}
        = \frac{\lambda_{i}}{\sqrt{1+y^2}} \frac{\partial}{\partial x} \pm \sqrt{2E_{i}- \lambda_{i}^2} \frac{\partial}{\partial y} + \frac{\lambda_{i} y}{\sqrt{1+y^2}} \frac{\partial}{\partial z},
    \end{equation}
    whose integral curves $\sigma_i(t) = (x_i(t), y_i(t), z_i(t))$ are given by
    \begin{equation}
    \begin{aligned}
        \dot x_i(t) = \frac{\lambda_{i}}{\sqrt{1+y(t)^2}},\,\dot y_i(t) = \pm \sqrt{2E_{i}- \lambda_{i}^2},\hbox{ and }
        \dot z_i(t) = \frac{\lambda_{i} y(t)}{\sqrt{1+y(t)^2}}.
    \end{aligned}
    \end{equation}
    Hence,
    \begin{equation}
    \begin{aligned}
        & x_i(t) 
        = \frac{\lambda_i}{A_i}\left[ \operatorname{arcsinh}\left(A_i t+y_{i,0}\right) -\operatorname{arcsinh}\left(y_{i,0}\right)\right] + x_{i,0},\\
        & y_i(t) = A_i t + y_{i,0},\\
        & z_i(t) = \frac{\lambda_i}{A_i} \left[\sqrt{(A_i t+y_{i,0})^2+1}- \sqrt{y_{i,0}^2+1}\right]+z_{i,0},
    \end{aligned}
    \end{equation}
    if $A_i\neq 0$, and
     \begin{equation}
    \begin{aligned}
        & x_i(t) = \frac{\lambda_{i}}{\sqrt{1+y_{i,0}^2}} + x_{i,0},\\
        & y_i(t) = y_{i,0},\\
        & z_i(t) = \frac{\lambda_{i} y_{i,0}}{\sqrt{1+y_{i,0}^2}}+z_{i,0},
    \end{aligned}
    \end{equation}
    if $A_i=0$,
    where $A_i=\pm \sqrt{2E_{i}- \lambda_{i}^2}$ and $x_{i+1,0}=x_{i}(t_{i}),\ y_{i+1,0}=y_{i}(t_{i})$ and $z_{i+1,0}=z_{i}(t_{i})$, for $i=0, \ldots, n$, with $t_0=0$ and $t_i$ the instant of the $i$-th impact. Fixing $\lambda_i,\ E_i$ and the initial conditions $(x_{0,0},y_{0,0},z_{0,0})$ determines any trajectory of the hybrid system. These trajectories coincide with the ones obtained by computing the integral curves of the Lagrangian nonholonomic vector field directly (see \cite[Example 4.3.4]{AnahorySimoes2021}).


    

\subsubsection{The generalized rigid body}

Consider a mechanical system with a Lie group as configuration space, namely $Q=G$. Let $\mathfrak g$ denote the Lie algebra of $G$ and $\mathfrak g^\ast$ its dual. Its Lagrangian is the left-invariant function $L\colon TG \simeq G\times \mathfrak{g} \to \R$ given by $L(g, v_g) = \ell(g^{-1} v_g)$, where $\ell\colon \mathfrak g \to \R$ is the reduced Lagrangian, defined by
\begin{equation}
    \ell(\xi) = \frac{1}{2} I_{ij} \xi^i \xi^j,
\end{equation}
for $\xi = (\xi^1, \ldots, \xi^n)\in \mathfrak g$, where $I_{ab}$ are the components of the (positive-definite and symmetric) inertia tensor $\mathbb I \colon \mathfrak g \to \mathfrak g^\ast$ (see for instance \cite{Bloch2005, Fernandez2014, Zenkov2000}).

The Legendre transform of $L$ is given by $\FF L \colon \left(g, \xi^i\right) \mapsto \left(g, I_{ij} \xi^j\right)$, so the Hamiltonian function $H\colon G \times \mathfrak g^\ast \to \R$ is
\begin{equation}
    H = \frac{1}{2} I^{ij} \eta_i\, \eta_j,
\end{equation}
where $I^{ij}$ are the components of the inverse of $\mathbb I$, and $\eta = (\eta_1, \ldots, \eta_n)\in \mathfrak g^\ast$.

The constrained generalized rigid body is subject to the left-invariant nonholonomic constraint
\begin{equation}
   \mathcal D_\mu = \left\{(g,\xi) \in G\times \mathfrak g \mid  \langle \mu, \xi \rangle = \mu_i\, \xi^i = 0\right\},
\end{equation}
where $\mu=(\mu_1, \ldots, \mu_n)$ is a fixed element of $\mathfrak g^\ast$ and $\langle \cdot, \cdot \rangle$ denotes the natural pairing between a Lie algebra and its dual. Then,
\begin{equation}
   \mathcal M_\mu = \FF L(\mathcal D_\mu) = \left\{(g,\eta) \in G\times \mathfrak g^\ast \mid  
   {\eta_i I^{ij} \mu_j=0} \right\}.
\end{equation}
{Since $\mathbb{I}$ is positive-definite, $\FF L=\flat_{\mathbb{I}}$ is a diffeomorphism and $\mathcal M_\mu$ is well-defined globally.}

A solution of the nonholonomic Hamilton--Jacobi problem is a one-form $\gamma\colon G \to G\times \mathfrak g^\ast,\ g\mapsto (g, \gamma_1(g), \ldots, \gamma_n (g))$ satisfying
\begin{equation}
\begin{aligned}
    & H\circ \gamma = \frac{1}{2} I^{ij} \gamma_i \gamma_j = E, \\
    & I^{ij} \gamma_i \mu_j = 0, \\
    & \restr{\dd \gamma}{\mathcal D\times \mathcal D} = 0.
\end{aligned}
\end{equation}

Hereinafter, let us consider $G=\mathrm{SO(3)}$. 
{Let $\{e_1, e_2, e_3\}$  be the canonical basis of $\mathfrak{so}(3)\simeq \R^3$, whose Lie brackets are
\begin{equation}
    [e_1, e_2] = e_3\, , \quad [e_1, e_3] = -e_2\, , \quad [e_2, e_3] = e_1\, ,
\end{equation}
and let $\{e^1, e^2, e^3\}$ be its dual basis.
For simplicity's sake, assume that
\begin{equation}
    \mathbb{I} = I e^1\otimes e^1 + I e^2\otimes e^2 + I e^3\otimes e^3\, ,
\end{equation}
and thus
\begin{equation}
    H(g,\eta) = \frac{1}{2I^2} \left(\eta_1^2 +  \eta_2^2 +  \eta_3^2 \right)\, .
\end{equation}}

The nonholonomic distribution is given by
\begin{equation}
   {\mathcal D_\mu = \left\{(g, \xi) \in \mathrm{SO(3)} \times \mathfrak{so}(3)\mid {\mu}_i \xi^i = 0 \right\}
    = \mathrm{span} \left\{{\mu}_2 e_1 -{\mu}_1 e_2,\ {\mu}_3 e_1 - {\mu}_1 e_3 \right\}.}      
\end{equation}
A solution of the Hamilton--Jacobi problem satisfies
{
\begin{equation} \label{HJ_eqs_SO3}
    \begin{aligned}
        & \gamma_1^2 +  \gamma_2^2 + \gamma_3^2 = 2E I^2, \\
        & \gamma_1 \mu_1 + \gamma_2 \mu_2 + \gamma_3 \mu_3  = 0, \\
        & \dd \gamma \left({\mu}_2 e_1 -{\mu}_1 e_2,\ {\mu}_3 e_1 - {\mu}_1 e_3 \right) = 0.
    \end{aligned}
    \end{equation}
}
Using the first two of the equations~\eqref{HJ_eqs_SO3} we can write
{
\begin{align}
    & \gamma_2 = \frac{\pm\mu_{3}\sqrt{ 2 E I^2 \left(\mu_{2} ^2+\mu_{3} ^2\right)-\gamma_{1} ^2 \left(\mu_{1} ^2+\mu_{2} ^2+\mu_{3} ^2\right)}-\mu_{1}  \mu_{2}  \gamma_{1}}{\mu_{2} ^2+\mu_{3} ^2},\\
    & \gamma_3 = \frac{\pm \mu_{2} \sqrt{ 2 E I^2\left(\mu_{2} ^2+\mu_{3} ^2\right)-\gamma_{1} ^2 \left(\mu_{1} ^2+\mu_{2} ^2+\mu_{3} ^2\right)}-\mu_{1} \mu_{3}   \gamma_{1}}{\mu_{2} ^2+\mu_{3} ^2}.
\end{align}
}
With the \textit{ansatz} $\gamma_1=\lambda_1$ for some $\lambda_1\in \R$, we obtain
{
\begin{equation}
    \gamma = \lambda_1 e^1 + \frac{\mu_{3} \lambda_2-\mu_{1}  \mu_{2}    \lambda_1 }{\mu_{2} ^2+\mu_{3} ^2} e^2 
    + \frac{\mu_{2}  \lambda_2 -\mu_{1} \mu_{3}  \lambda_1}{\mu_{2} ^2+\mu_{3} ^2} e^3,
\end{equation}
where $\lambda_2 = \pm \sqrt{ 2 E I^2 \left(\mu_{2} ^2+\mu_{3} ^2\right)-\lambda_{1} ^2 \left(\mu_{1} ^2+\mu_{2} ^2+\mu_{3} ^2\right)}$. Moreover, using that 
\begin{equation}
    \dd \gamma (X, Y) = \liedv{X} \contr{Y} \gamma + \contr{[X,Y]} \gamma - \liedv{X} \contr{Y} \gamma
\end{equation}
for $X= {\mu}_2 e_1 -{\mu}_1 e_2$ and $Y = {\mu}_3 e_1 - {\mu}_1 e_3$, one can verify that the third of the equations~\eqref{HJ_eqs_SO3} is satisfied.
}

Let us now use the Euler angles $(\alpha, \beta, \varphi)$ as a coordinate system for $\mathrm{SO}(3)$, i.e., the action of an element $g\in\mathrm{SO}(3)$ with coordinates $(\alpha, \beta, \varphi)$ on a point $(x,y, z)\in \R^3$ by rotations is given by the matrix multiplication
\begin{equation}
    \begin{pmatrix}
        \cos \beta \cos \varphi    
        & \sin \alpha \sin \beta  \cos \varphi  -  \cos \alpha  \sin \varphi
        & \cos \alpha \sin \beta \cos\varphi +  \sin \alpha \sin \varphi
        \\
         \cos \beta  \sin \varphi 
        & \sin \alpha \sin \beta  \sin \varphi +  \cos \alpha \cos \varphi
        & \cos \alpha \sin \beta \sin\varphi - \sin \alpha  \cos \varphi
        \\ 
        -\sin \beta 
        & \sin \alpha \cos \beta 
        & \cos \alpha \cos \beta
    \end{pmatrix}
    \begin{pmatrix}
        x \\ y \\ z
    \end{pmatrix}.
\end{equation}
Define the switching surface as the codimension-1 submanifold $S_H$ of $\mathrm{SO(3)}\times \mathfrak{so}(3)^\ast$ given by
\begin{equation}
    S_H = \left\{ \left(\alpha, \beta, \varphi, \eta_1, \eta_2, \eta_3\right)\in \mathrm{SO(3)}\times \mathfrak{so}(3)^\ast\mid \alpha=0 \right\},
\end{equation}
and the impact map $\Delta_H\colon S_H \to \mathrm{SO(3)}\times \mathfrak{so}(3)^\ast$ by
\begin{equation}
    \Delta_H \colon 
    \left(0, \beta, \varphi, \eta_1, \eta_2, \eta_3\right) \mapsto
    \left(0, \beta, \varphi, \varepsilon \eta_1, \eta_2,  \eta_3 \right),
\end{equation}
for some constant $\varepsilon$. 
{Note that this impact map does not preserve the nonholonomic distribution. In other words, it is possible that after the impact the momenta no longer satisfy the nonholonomic constraints. This type of impact maps could be employed to model dynamical systems whose constraints are modified. For instance, a ball that switches from rolling without sliding to moving freely. Nevertheless, the present example is just an illustrative academic example that pretends to show the theory developed. In future works we plan to study more realistic examples of nonholonomic hybrid systems.}


Let $\gamma^-$ and $\gamma^+$ denote the solutions to the Hamilton--Jacobi equation before and after the impact, respectively, where
{
\begin{equation}
    \gamma^\pm = \lambda_1^\pm e^1 + \frac{\mu_{3} \lambda_2^\pm-\mu_{1}  \mu_{2}    \lambda_1^\pm }{\mu_{2} ^2+\mu_{3} ^2} e^2 
    + \frac{\mu_{2}  \lambda_2^\pm -\mu_{1} \mu_{3}  \lambda_1^\pm}{\mu_{2} ^2+\mu_{3} ^2} e^3,
\end{equation}
}
Then, $\gamma^+$ and $\gamma^-$ are $\Delta_H$-related if and only if
{
\begin{align}
    & \lambda_1^+ = \varepsilon \lambda_1^-,\\
    & \mu_{3}\lambda_2^+ - \mu_{1}  \mu_{2}   \lambda_1^+ = \mu_{3}\lambda_2^- - \mu_{1}  \mu_{2}   \lambda_1^- ,\\
    & \mu_{2}  \lambda_2^+ -\mu_{1} \mu_{3}  \lambda_1^+ = \mu_{2}  \lambda_2^- -\mu_{1} \mu_{3}  \lambda_1^-,
\end{align}
}
that is,
{
\begin{align}
    & \lambda_1^+ = \varepsilon \lambda_1^-,\\
    & \lambda_2^+ = \lambda_2^- + (\varepsilon-1) \frac{\mu_1\mu_2}{\mu_3}\lambda_1^-, \\
   & \lambda_2^+ = \lambda_2^- + (\varepsilon-1) 
   \frac{\mu_1\mu_3}{\mu_2}\lambda_1^- ,
\end{align}
}
which has solutions if $\mu_3=\pm \mu_2$ or if $\varepsilon=1$. We conclude that $\hybrid_{\mathrm{nh}} = (T^\ast Q, \mathcal M_\mu, {S_H}, X_{H}^{\mathrm{nh}}, \Delta_H)$ is an integrable nonholonomic hybrid system for $\mu=(\mu_1, \mu_2, \pm \mu_2)\in \mathfrak{so}(3)^\ast$ or for $\Delta_H$ the canonical inclusion.

In the case that $\mu=(\mu_1, \mu_2, \mu_2)\in \mathfrak{so}(3)^\ast$, the nonholonomic distribution is 
\begin{equation}
    \mathcal D_\mu = \left\{(g, \xi) \in \mathrm{SO(3)} \times \mathfrak{so}(3)\mid a_i \xi^i = 0 \right\}
    = \mathrm{span} \left\{a_2 e_1 -a_1 e_2,\ a_2 e_1 - a_1 e_3 \right\}.
\end{equation}
The solutions to the Hamilton--Jacobi equation are
{
\begin{equation}
    \gamma^\pm = \lambda_1^\pm e^1 + \frac{\lambda_2^\pm-\mu_{1}    \lambda_1^\pm }{2\mu_{2}} e^2 
    + \frac{\lambda_2^\pm -\mu_{1} \lambda_1^\pm}{2\mu_{2}} e^3\, .
\end{equation}
}

{
\section{Conclusions and further work}\label{sec:conclusions}

In the present paper, following the seminal work by Clark \cite{Clark2020}, we have developed a Hamilton--Jacobi theory for hybrid Hamiltonian systems, forced hybrid Hamilton systems and nonholonomic hybrid systems. We have illustrated our results with several examples. This theory is based on employing the geometric Hamilton--Jacobi theory for Hamiltonian, forced Hamiltonian or nonholonomic systems on the continuous parts and enforcing the solutions between impacts to ``glue adequately''.

An interesting future work would be to develop a more general notion of hybrid integrability which does not depend on the integrability of the underlying continuous system. Additionally, we would like to study systems experiencing Zeno effect, that is, impacts are no longer discrete. As it has been mentioned throughout  the paper, another open problem is the existence of action-angle like coordinates on forced Hamiltonian systems and nonholonomic systems. Furthermore, we would like to delve on nonholonomic hybrid systems, considering impact maps which may or may not preserve the nonholonomic distribution. This could be employed, for instance, for modeling systems whose dynamics switch from a motion with nonholonomic constraints to a free motion, e.g., a ball which first rolls without sliding and then begins to slide. 
}

\section*{Data availability statement}

No new data were created or analyzed in this study.

\section*{Acknowledgements}
The authors are thankful to the referees for their insightful comments and constructive feedback which have improved the overall quality of the paper.
L.~Colombo, M.~de León and A.~López-Gordón acknowledge financial support from Grants PID2019-106715GB-C21, PID2022-137909NB-C21 and RED2022-134301-T funded by MCIN/AEI/ 10.13039/501100011033. M.~E.~Eyrea Irazú received the support of CONICET. M.~de León and A.~López-Gordón also recieved support from the Grant CEX2019-000904-S funded by MCIN/AEI/ 10.13039/501100011033. A.~López-Gordón would also like to thank MCIN for the predoctoral contract PRE2020-093814.


\let\emph\oldemph
\printbibliography

\end{document}